\DeclareSymbolFont{AMSb}{U}{msb}{m}{n}
\documentclass[reqno, a4paper, noamsfonts, 11pt]{amsart}
\usepackage[utf8]{inputenc}
\usepackage[T1]{fontenc}
\usepackage[english]{babel}
\usepackage[usenames,dvipsnames]{xcolor}
\usepackage{graphicx}
\usepackage{hyperref}
\usepackage{dsfont}
\usepackage{lmodern}
\usepackage{microtype}
\usepackage{mathtools}
\mathtoolsset{centercolon,showonlyrefs,showmanualtags}
\usepackage{chemfig}
\usepackage{nicefrac}
\definecolor{darkblue}{rgb}{0,0,0.65}
\definecolor{darkred}{rgb}{0.85,0,0}
\hypersetup{
    pdftitle={The exchange-driven growth model: basic properties and longtime behavior},    
    pdfauthor={André Schlichting},
    pdfsubject={},   
    pdfkeywords={},
    colorlinks=true,       
    linkcolor=darkred,          
    citecolor=darkblue,        
    filecolor=blue,      
    urlcolor=blue           
}

\usepackage[text={33pc,605pt},centering]{geometry}    
\renewcommand{\familydefault}{bch}
\usepackage[charter,expert]{mathdesign}
\newcommand{\weakto}{\xrightharpoonup{*}}
\DeclarePairedDelimiter{\abs}{\lvert}{\rvert}
\DeclarePairedDelimiter{\norm}{\lVert}{\rVert}
\DeclarePairedDelimiter{\bra}{(}{)}
\DeclarePairedDelimiter{\pra}{[}{]}
\DeclarePairedDelimiter{\set}{\{}{\}}

\DeclareMathAlphabet{\mathup}{OT1}{\familydefault}{m}{n}
\newcommand{\dx}[1]{\mathop{}\!\mathup{d} #1}
\newcommand{\pderiv}[3][]{\frac{\mathop{}\!\mathup{d}^{#1} #2}{\mathop{}\!\mathup{d} #3^{#1}}}
\newcommand{\eps}{\ensuremath{\varepsilon}}
\newcommand{\N}{\mathds{N}}
\newcommand{\R}{\mathds{R}}

\newcommand{\cA}{\ensuremath{\mathcal A}}
\newcommand{\cB}{\ensuremath{\mathcal B}}
\newcommand{\cC}{\ensuremath{\mathcal C}}
\newcommand{\cD}{\ensuremath{\mathcal D}}

\newcommand{\cF}{\ensuremath{\mathcal F}}

\newcommand{\cH}{\ensuremath{\mathcal H}}

\newcommand{\cK}{\ensuremath{\mathcal K}}

\newcommand{\cO}{\ensuremath{\mathcal O}}
\newcommand{\cP}{\ensuremath{\mathcal P}}

\newcommand{\cS}{\ensuremath{\mathcal S}}

\newcommand{\cX}{\ensuremath{\mathcal X}}
\newcommand{\cY}{\ensuremath{\mathcal Y}}

\DeclareMathOperator{\Boltz}{B}

\newcommand{\upd}{\ensuremath{\mathup d}}
\newcommand{\overbar}[1]{\mkern 1.5mu\overline{\mkern-1.5mu#1\mkern-1.5mu}\mkern 1.5mu}
\newtheorem{thm}{Theorem}[section]
\newtheorem{cor}[thm]{Corollary}
\newtheorem{lem}[thm]{Lemma}
\newtheorem{prop}[thm]{Proposition}

\theoremstyle{definition}
\newtheorem{defn}[thm]{Definition}
\newtheorem{ex}[thm]{Example}
\newtheorem{ass}[thm]{Assumption}
\theoremstyle{remark}
\newtheorem{rem}[thm]{Remark}

\numberwithin{equation}{section}
\begin{document}

\title[The exchange-driven growth model: basic properties and longtime behavior]{The exchange-driven growth model: \\
basic properties and longtime behavior}

\author{André Schlichting}
\address{Institut für Angewandte Mathematik, Universität Bonn} 
\email{schlichting@iam.uni-bonn.de}

\subjclass[2010]{%
Primary: 34A35; 
Secondary: 
34G20, 
37B25, 
37D35, 
82C05, 
82C26  
}

\keywords{%
convergence to equilibrium,
exchange-driven growth, 
entropy method,
mean-field equation,
zero-range process}

\date{\today}

\begin{abstract}
  The exchange-driven growth model describes a process in which pairs of clusters interact through the exchange of single monomers. The rate of exchange is given by an interaction kernel $K$ which depends on the size of the two interacting clusters. Well-posedness of the model is established for kernels growing at most linearly and arbitrary initial data. 
    
  The longtime behavior is established under a detailed balance condition on the kernel. The total mass density $\varrho$, determined by the initial data, acts as an order parameter, in which the system shows a phase transition. There is a critical value $\varrho_c\in (0,\infty]$ characterized by the rate kernel. For $\varrho \leq \varrho_c$, there exists a unique equilibrium state $\omega^\varrho$ and the solution converges strongly to $\omega^\varrho$. If $\varrho > \varrho_c$ the solution converges only weakly to $\omega^{\varrho_c}$. In particular, the excess $\varrho - \varrho_c$ gets lost due to the formation of larger and larger clusters. In this regard, the model behaves similarly to the Becker-Döring equation.

  The main ingredient for the longtime behavior is the free energy acting as Lyapunov function for the evolution. It is also the driving functional for a gradient flow structure of the system under the detailed balance condition.
\end{abstract}

\maketitle

\section{Introduction}

\subsection{Model}

The exchange-driven growth model describes a process in which pairs of clusters consisting of an integer number of monomers can grow or shrink only by the exchange of single monomers~\cite{NK2003}.
Although this process is not necessarily realized by chemical kinematics, it is convenient to be interpreted as a reaction network of the form
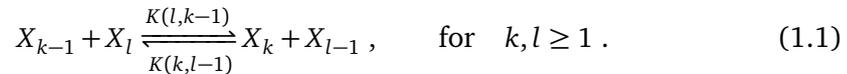
\begin{equation}\label{e:EG:ChemReact}
  X_{k-1} + X_{l} \xrightleftharpoons[K(k,l-1)]{K(l,k-1)} X_{k} + X_{l-1}\ , \qquad \text{for}\quad k, l \geq 1 \ . 
\end{equation}
The clusters of size $k\geq 1$ are denoted by $X_k$. Additionally, the variable~$X_0$ represents empty volume. Here, the kernel  $K(k,l-1)$ encodes the rate of the exchange of a single monomer from a cluster of size $k$ to a cluster of size $l-1$. Here and in the following the notation $k\geq 1$ means $k\in \N=\set{1,2,\dots}$ and $l\geq 0$ denotes $l\in \N_0 = \N \cup\set{0}$.

The concentrations of $X_k$ in~\eqref{e:EG:ChemReact} are denoted by $(c_k)_{k\geq 0}$ and satisfy for $k\geq 0$ the reaction rate equation formally obtained from~\eqref{e:EG:ChemReact} by mass-action kinetics
\begin{equation}\label{e:master}
\begin{aligned}
\dot c_k \ \ = \ \ \ &\sum_{l\geq 1} K(l,k-1)c_l c_{k-1} - \sum_{l\geq 1} K(k,l-1) c_{k} c_{l-1}  \\
 - &\sum_{l\geq 1} K(l,k)c_l c_k + \sum_{l\geq 1} K(k+1,l-1)c_{k+1} c_{l-1} \ , \qquad\text{ for } k\geq 0\ , 
\end{aligned}
\end{equation}
with $c_{-1}\equiv 0$ set for convenience.

The model~\eqref{e:master} is applied to social phenomena like migration~\cite{Ke2},
population dynamics~\cite{Leyvraz2} and wealth exchange~\cite{Isp}. Similar driving
mechanisms are found in diverse phenomena at
contrasting scales from microscopic level polymerization processes~\cite{DoiEdwards1988}, to cloud~\cite{Drake} and galaxy formation mechanisms at huge scales, as well as in statistical physics~\cite{Naim-book}.

Moreover, the model~\eqref{e:master} also arises as the mean-field limit of a class of interacting particle systems that include extensively studied models of nonequilibrium statistical physics like the zero-range processes \cite{Godrec,Stef3,Godrec2,Beltran,JG2017}, and more general misanthrope processes \cite{Waclaw,Stef4,Colm}.

This work extends and complements the basic mathematical analysis of~\cite{emre} in two ways. Firstly, it improves parts of the well-posedness results making them unconditioned on the initial data. Secondly, the new main result is the qualitative longtime behavior for kernels with sublinear growth. In addition, the aim is to stress the observation that this model is a natural generalization of the Becker-Döring model~\cite{BD35} (see Example~\ref{ex:BeckerDoering}) and resembles very much of its qualitative behavior.

The chemical reaction representation~\eqref{e:EG:ChemReact} gives rise to two conservation laws. Firstly, on each side of the reaction there are two clusters, or a cluster and empty volume, which leads to the conservation of the total number of clusters and empty volume. Due to each reaction performing an exchange of a single monomer, no mass is generated nor destroy, which gives the conservation of the total number of monomers. On the level of the densities $(c_k)$, these two conservation laws take the form of
\begin{equation}\label{e:ConverationLaws}
  M_0 = \sum_{k\geq 0} c_k \qquad\text{and}\qquad \varrho = \sum_{k \geq 1} k \, c_k \ . 
\end{equation}
After rigorously establishing both conservation laws (Corollary~\ref{cor:convervation_laws}), the zeroth moment can be fixed to be $M_0 = 1$. This allows to interpret \eqref{e:master} as the master equation for a nonlinear continuous-time birth-death chain on $\N_0$ with distribution~$\set{c_k(t)}_{k\geq 0}$. This chain is nonlinear since the birth and death rates
\begin{equation}\label{e:bdrates}
A_{k-1}[c]=\sum_{l\geq 1} K(l,k-1) \, c_l \quad\text{and}\quad B_k[c] = \sum_{l\geq 0} K(k,l) \, c_l  \quad\text{for } k\geq 1
\end{equation}
depend on the distribution $c(t)$, where in the following $B_0[c]=0$ is set. For the mathematical analysis this interpretation turns out to be useful. 

It will be convenient to introduce certain fluxes, which allow to rewrite the system~\eqref{e:master} in a compact form. 
There are the unidirectional fluxes per reaction in~\eqref{e:EG:ChemReact} given by $j_{l,k-1}[c] = K(l,k-1) \, c_l \, c_{k-1}$. Their difference $j_{l,k-1}[c] - j_{k,l-1}[c]$ is the net flux per reaction $(k-1,l)\to (k,l-1)$. The summation over all possible reaction partners $l$ leads to the definition of the net flux from a cluster of size $k-1$ to one of size $k$
\begin{equation}\label{e:def:net:flux}
  J_{k-1}[c] = \sum_{l\geq 1} \bra*{ j_{l,k-1}[c] - j_{k,l-1}[c]} = A_{k-1}[c] c_{k-1} - B_{k}[c] c_{k} \quad\text{ for } \quad k\geq 1 \ ,
\end{equation}
By having introduced the fluxes~\eqref{e:def:net:flux}, the rate equation~\eqref{e:master} shortens to
\begin{equation}\label{e:master:flux}
  \dot c_k = J_{k-1}[c] - J_{k}[c] \ , \qquad\qquad k\geq 0 \  ,
\end{equation}
with $J_{-1}[c]\equiv 0$ set for convenience for all $c$.
\begin{ex}[Becker-Döring type model]\label{ex:BeckerDoering}
As also observed in~\cite{emre}, a particular case of this model is given by setting 
\begin{equation}\label{ass:rates:BD}
  l\geq 0: K(1,l) = a_l \ , \quad k\geq 1: K(k,0) = b_k \quad \text{and for}\quad k\geq 2, l\geq 1: K(k,l)=0  \ . 
\end{equation}
This choice simplifies the chemical reaction network~\eqref{e:EG:ChemReact} to 
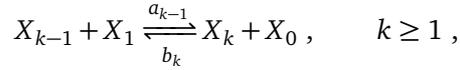
\begin{equation}\label{e:BD:ChemReact}
  X_{k-1} + X_{1} \xrightleftharpoons[b_k]{a_{k-1}} X_{k} + X_{0} \ , \qquad  k \geq 1 \ ,
\end{equation}
which corresponds to a model very close to the Becker-Döring model~\cite{BD35}. The main difference to the Becker-Döring model
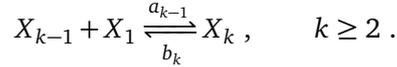
\begin{equation}\label{e:classicBD:ChemReact}
  X_{k-1} + X_{1} \xrightleftharpoons[b_k]{a_{k-1}} X_{k} \ , \qquad  k \geq 2 \ .
\end{equation}
is the additional variable $X_0$ corresponding to empty volume. Taking into account~$X_0$ gives rise to the first conservation in~\eqref{e:ConverationLaws}. The consequence is the fragmentation flux $b_k \, c_0 \, c_k$ becoming nonlinear taking finite volume effects into account. 
\end{ex}

\subsection{Main results: Well-posedness and convergence to equilibrium}

In view of the two conservation laws~\eqref{e:ConverationLaws}, the equation \eqref{e:master} is studied in the normed vector space
\begin{equation}\label{e:def:X}
 \cX=\set{ c \in \ell^1(\N_0) \ :\  \norm{c} < \infty} \quad\text{with}\quad \norm{c} = \sum_{l\geq 0} (1+l) \, \abs{c_l} \ .
\end{equation}
Moreover, by Theorem~\ref{thm:stability} it is shown that solutions to~\eqref{e:master} are nonnegative for nonnegative initial data and remain in the cone of nonnegative densities
\begin{equation}\label{e:def:Xplus}
 \cX^+ =\set{ c \in \cX \ :\ c_l \geq 0}  \ . 
\end{equation}
Additionally, by the two conservation laws~\eqref{e:ConverationLaws} rigorously established in Corollary~\ref{cor:convervation_laws}, the total number density can be normalized to $1$ and hence the state space for the evolution is the subspace of probability densities on~$\N_0$ with the first moment fixed by the parameter $\varrho\geq 0$
\begin{equation}\label{e:def:states}
 \cP^\varrho = \set[\bigg]{ c\in \ell^1(\N_0) \ : \ c_l \geq 0 \ ,\  \sum_{l\geq 0} c_l = 1 \ ,\  \sum_{l\geq 1} l \, c_l = \varrho}\subset \cX^+ \ . 
\end{equation}
To establish relative compactness and tightness the union of the above spaces is used
\begin{equation}\label{e:def:ball}
  \cB^\varrho = \set[\bigg]{ c\in \ell^1(\N_0) \ :\  c_l \geq 0 \ ,\  \sum_{l\geq 0} c_l = 1 \ ,\  \sum_{l\geq 1} l \, c_l \leq \varrho } \subset \cX^+ \ . 
\end{equation}
\begin{ass}[Well-posedness]\label{ass:well-posedness}
  The kernel $K: \N\times \N_0 \to [0,\infty)$ is supposed to have at most linear growth by either assuming for some $C_K \in (0,\infty)$
   \begin{equation}\label{e:ass:K1}
   0 \leq K(k,l-1) \leq C_K \, k\, l \qquad\text{ for } k,l \geq 1 \ . \tag{K$_1$}
  \end{equation}
  or the stronger assumption for all $k,l\geq 1$
  \begin{equation}\label{e:ass:K2}
    \abs{K(l,k)-K(l,k-1)} \leq C_K \, l  \quad\text{and}\quad \abs{K(l+1,k-1)-K(l,k-1)}\leq C_K \, k  \ .  \tag{K$_2$}
  \end{equation}
\end{ass}
The analysis of well-posedness is restricted to the case of kernels growing at most linearly for two reasons. Firstly, the well-posedness theory for arbitrary initial data in $\cP^\varrho$ is almost complete, except for a small gap between the Assumption~\eqref{e:ass:K1} for existence and the slightly stronger one~\eqref{e:ass:K2} for uniqueness. Secondly, the longtime behavior is still very interesting, since the system can exhibit a phase transition related to the ergodic behavior of solutions shown in Theorem~\ref{thm:longtime} below.
Cases with faster than linear growing kernels are treated in~\cite{emre} and examples with gelation are found. 
\begin{thm}[Well-posedness]\label{thm:well-posed}
  Suppose Assumption~\eqref{e:ass:K1} holds. Then for any $\varrho>0$ and $\bar c\in \cP^\varrho$ there exists a solution $\bra[\big]{c(t)}_{t\geq 0}$ to~\eqref{e:master} with initial datum $c(0)=\bar c$ in the sense of Definition~\ref{def:sol}. If also~\eqref{e:ass:K2} holds, this solution is unique. Moreover, in the latter case, the solutions constitute a semigroup on $\cP^\varrho$ (Definition~\ref{def:flows}).
\end{thm}
The convergence to equilibrium is proven under a detailed balance condition. This case is already interesting, since it shows a phase transition in the order parameter $\varrho$. The existence of detailed balance states will turn out to be equivalent to some additional assumption on the rates. This assumption~\eqref{e:ass:BDA} below was already obtained for the stochastic particle system in~\cite[(5.3)]{CGR18} and used to show that stationary states are of product form. 

Besides the detailed balance condition, more information on the kernel $K(k,l-1)$ is needed, especially on asymptotic growth and regularity properties for $k,l$ large. Moreover, the proof of the relative compactness of solutions to~\eqref{e:master} in $\cP^\varrho$ is restricted to (strictly) sublinear growth rates, since only in this case the nonlinear birth and death rates~\eqref{e:bdrates} are controlled by a tightness argument.
\begin{ass}[Longtime behavior]\label{ass:BDA}
The rate kernel $K: \N \times \N_0 \to [0,\infty)$ satisfies the \emph{Becker-Döring assumption}, that is for all $k,l\geq 1$ it holds $K(k,l-1)>0$ and 
\begin{equation}\label{e:ass:BDA}
 \frac{K(k,l-1)}{K(l,k-1)} = \frac{K(k,0) \, K(1,l-1)}{K(l,0) \, K(1,k-1)}  \ . \tag{BDA} 
\end{equation}
The kernel $K$ satisfies
\begin{equation}\label{e:ass:Kc}
  \lim_{k\to \infty} \frac{K(k,0)}{K(1,k-1)} = \phi_c \in (0,\infty] \ . \tag{K$_c$}
\end{equation}
The kernel $K$ satisfies~\eqref{e:ass:K2} and the following continuity at infinity
\begin{equation}\label{e:ass:K3}
  \lim_{k \to \infty} \frac{K(l,k)}{K(l,k-1)} = 1 \quad\text{and}\quad \lim_{k \to \infty} \frac{K(k,l-1)}{K(k-1,l-1)} = 1 \quad\text{uniformly in $l\geq 1$} \ .  \tag{K$_3$} 
\end{equation}
Moreover, for three sublinear increasing sequences $(a_k)_{k\geq 0}$, $(b_k)_{k\geq 1}$ and $(d_k)_{k\geq 1}$ there exists a constant $C_K\geq 1$ such that for all $k,l\geq 1$ it holds
\begin{equation}\label{e:ass:K4}
  C_K^{-1} \, a_{l-1} \leq \abs{K(k,l-1)} \leq C_K \, d_k \, a_{l-1}  \quad\text{and}\quad C_K^{-1} \, b_k \leq \abs{K(k,l-1)} \leq C_K \, b_k \, l \ ,  \tag{K$_4$}
\end{equation}
Hereby, a sequence $(d_k)_{k \geq 1}$ is called sublinear, if $\lim_{k \to \infty} \frac{d_k}{k} = 0$.
\end{ass}
\begin{rem}
  For the majority of the proofs, it seems possible to relax the above strict positivity assumption. That is $K(k,0)>0$ and $K(1,k-1)>0$ for $k\geq 1$ and~\eqref{e:ass:BDA} holds only on the support of $K$ implying that with $K(k,l-1)>0$ also $K(l,k-1)>0$. 
  
  The presentation is restricted to the case of positive rates, but the discussion below applies with minor obvious changes to Example~\ref{ex:BeckerDoering} satisfying~\eqref{e:ass:BDA} in the above sense.
  
  The Assumption~\ref{ass:BDA} includes~\eqref{e:ass:K2}, which by Theorem~\ref{thm:well-posed} provides a unique global solution to~\eqref{e:master} for which the longtime behavior is established below.
\end{rem}
\begin{ex}
  A family of kernels satisfying Assumption~\eqref{e:ass:BDA} is given by the \emph{modulated separable} kernel 
   \begin{equation}\label{e:kernel:separable:modulated}
     K(k,l-1) = b_k a_{l-1} S(k,l-1) \ , \qquad\text{ for } k,l\geq 1 \ ,
   \end{equation} 
  where $S(k,l-1)$ is positive and symmetric $S(k,l-1) = S(l,k-1) > 0$ for $k,l\geq 1$. A particular family of kernels, called \emph{separable} kernels, is obtained for $S(\cdot,\cdot)\equiv 1$.
  
  Many important mean-field limits of misanthrope-type stochastic particle systems~\cite{Cocozza-Thivent1985} have rate kernels of the following general form
  \[
     K(k,l-1) = k^{\alpha}\, \bra*{ a + \frac{q}{k^\gamma}} \, \bra*{ l^\delta + d} \, s_{k+l} \ , \qquad\text{ for } k,l\geq 1 \ , 
  \]
  where $\delta \in [0,1)$, $\alpha \in [\delta,1)$, $a>0$, $q> 0$, $\gamma \geq 0$,  and $\delta\, d > 0$. Moreover, $\set*{s_{r}}_{r\geq 1}$ satisfied $0 < s_* \leq s_r \leq s^* < \infty$ for all $r\geq 1$ and $s_r \to \bar s$ as $s\to \infty$. These family of kernels is compatible with the modulated separable kernel~\eqref{e:kernel:separable:modulated} and satisfies Assumption~\ref{ass:BDA} for a suitable range of parameters. First, $\phi_c$ in~\eqref{e:ass:Kc} is given by
  \begin{equation*}
    \frac{K(k,0)}{K(1,k-1)} = \frac{k^\alpha\bra*{ a + q k^{-\gamma}}(d+1) s_{k+1}}{(a+q) (k^\delta+d) s_{k+1}} \xrightarrow{k\to \infty} \phi_c= 
    \begin{cases}
      0 , & \alpha < \delta \\
      +\infty , & \alpha > \delta  \\
      \frac{a(d+1)}{a+q} & \alpha=\delta >0 , \gamma>0\\
      d+1 , & \alpha=\delta > 0 , \gamma= 0  \\
      \frac{a}{a+q} , & \alpha= \delta = 0, \gamma>0  \\
      1 , & \alpha= \gamma = \delta= 0  
    \end{cases} \ .
  \end{equation*}
  Hence, $\alpha\geq \delta$ is necessary for~\eqref{e:ass:Kc} with $\phi_c>0$. It is easy to check, that also~\eqref{e:ass:K2} and~\eqref{e:ass:K3} is satisfied, since $K$ is modeled as a rational function. Moreover, estimates in~\eqref{e:ass:K4} follow by choosing $a_{l-1} = l^\delta + d$, $b_k=k^\alpha(a+ q k^{-\gamma})$ and $d_k = k^{\alpha}$, where these sequences are sublinear under the present assumptions. 
\end{ex}
The Assumption~\eqref{e:ass:BDA} is called the \emph{Becker-Döring assumption} because, instead of a direct exchange of a single monomer from an $l$-cluster to a $(k-1)$-cluster, the jump is achieved through a jump to empty volume. This is visualized by the following network, where two intermediate reactions involving the monomers $X_1$ and empty volume $X_0$ with the other occurring rates in~\eqref{e:ass:BDA} are added
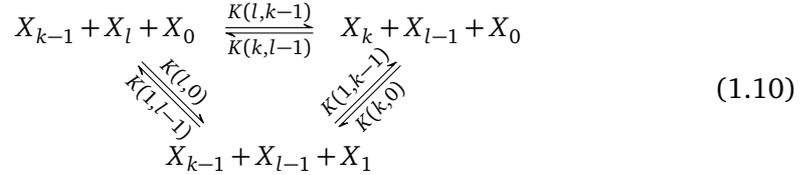
\begin{equation}\label{e:BDA:ChemReact}\raisetag{16pt}
\begin{aligned}
\schemestart
  $X_{k-1} + X_{l-1} + X_{1}$
  \arrow(.165--){<=>[$\scriptstyle K(l,0)$][$\scriptstyle K(1,l-1)$]}[135]
  $X_{k-1} + X_{l} + X_0$
  \arrow{<=>[$\scriptstyle K(l,k-1)$][$\scriptstyle K(k,l-1)$]}
  $X_{k}+X_{l-1}+X_0$
  \arrow{<=>[$\scriptstyle K(1,k-1)$][$\scriptstyle K(k,0)$]}[225]
\schemestop
\end{aligned}
\end{equation}
From the chemical network representation~\eqref{e:BDA:ChemReact}, the Assumption \eqref{e:ass:BDA} rewritten in the form
\[
   K(k,l-1) \, K(1,k-1) \, K(l,0) = K(l,k-1) \, K(1,l-1) \, K(k,0) 
\]
can be viewed as a curl-free property of the rate kernel on the reaction graph. 

For this reason it is not surprising that under Assumption~\eqref{e:ass:BDA}, there exists a chemical potential $(Q_k)_{k\geq 0}$ defined by
\begin{equation}\label{def:Q}
  Q_0 = 1 \qquad\text{and}\qquad Q_l = \prod_{k=1}^{l} \frac{K(1,k-1)}{K(k,0)} \ .
\end{equation}
Note that Assumption~\eqref{e:ass:Kc} implies that
\begin{equation}\label{e:Qk:phi_c}
  \lim_{k\to \infty} Q_k^{1/k} = \phi_c^{-1} \qquad\text{with the convention } \phi_c^{-1} = 0 \text{ when } \phi_c = \infty \ . 
\end{equation}
Thanks to~\eqref{e:ass:BDA}, the chemical potential $(Q_k)_{k\geq 0}$ satisfies the detailed balance condition
\begin{equation}\label{e:Q:DBC}
  K(k,l-1) \, Q_k \, Q_{l-1} = K(l,k-1) \, Q_{l} \, Q_{k-1} \qquad\text{ for } k,l\geq 1 \tag{DBC}
\end{equation}
and it is easily verified that~\eqref{e:Q:DBC} is actually equivalent to Assumption~\eqref{e:ass:BDA}. The two conversation laws~\eqref{e:ConverationLaws} are also encoded in~\eqref{e:Q:DBC}, since $(Z^{-1} \phi^k Q_k)_{k\geq 0}$ satisfies~\eqref{e:Q:DBC} for any $Z,\phi>0$. 

This observation is used to search for equilibrium states in $\cP^\varrho$ with $\varrho>0$. The Assumption~\eqref{e:ass:Kc} allows to define the partition sum $Z(\phi) \in [0,\infty)$ for $\phi\in [0,\phi_c)$ by
\begin{equation}\label{e:def:Z}
 Z(\phi) = \sum_{l\geq 0} \phi^l Q_l \in [0, \infty)\ .
\end{equation}
For $\phi \in [0,\phi_c)$, the normalized equilibrium states $\omega(\phi)$ are given by
\begin{equation}\label{e:def:omega}
 \omega_{l}(\phi) = Z(\phi)^{-1} \, \phi^l \, Q_l \qquad\text{for } l\geq 0 \ .
\end{equation}
The critical equilibrium density $\varrho_c \in (0,\infty]$ is defined by
\begin{equation}\label{def:varrhos}
  \varrho_c = \limsup_{\phi \uparrow \phi_c} Z(\phi)^{-1} \sum_{l\geq 1} l \phi^l Q_l \ .
\end{equation}
For $\varrho < \infty$ with $0 \leq \varrho \leq \varrho_c$, there exists a unique $\phi=\phi(\varrho)\in [0,\phi_c]$ such that 
\begin{equation}\label{e:def:phi:varrho}
  \sum_{l\geq 1} l \omega_l(\phi) = Z(\phi)^{-1} \sum_{l\geq 1} l \phi^l Q_l = \varrho \ .
\end{equation}
Indeed, the Jensen inequality implies the strict monotonicity property for $\phi\in (0,\phi_c)$
\begin{align}
 \phi \pderiv{}{\phi} \frac{\sum_{l\geq 1} l \phi^l Q_l}{\sum_{l\geq 0} \phi^l Q_l} &= \frac{\sum_{l\geq 1} l^2 \phi^l Q_l}{\sum_{l\geq 0} \phi^l Q_l} - \bra*{\frac{\sum_{l\geq 1} l \phi^l Q_l}{\sum_{l\geq 0}\phi^l Q_l} }^2 \notag \\
 &= \sum_{l\geq 1} l^2 \omega_l(\phi) - \bra*{\sum_{l\geq 1} l \omega_l(\phi)}^2 > 0 \ . \label{e:phi:strict:monotone}
\end{align}
Moreover, in the case $\varrho_c<\infty$, it follows also $\phi_c< \infty$. Indeed, suppose $\phi_c=\infty$ for $\varrho_c<\infty$, then for any $L\geq 1$ and any $\phi\geq 1$ by using also the above established monotonicity, it follows
\begin{align*}
  \varrho_c &\geq \frac{\sum_{l\geq 1} l \, \phi^l Q_l}{\sum_{l\geq 0} \phi^l Q_l} \geq  \frac{\sum_{l\geq L} l \, \phi^l Q_l}{\sum_{l = 0}^{L-1} \phi^l Q_l+ \sum_{l\geq L} \phi^l Q_l}
  \geq 
  \frac{L}{\frac{\sum_{l=0}^{L-1} \phi^l Q_l}{\sum_{l\geq L}^\infty l \, \phi^l Q_l} + 1} \\
  &\geq \frac{L}{\frac{(L-1) \phi^{L-1} \sup_{1\leq l \leq L} Q_l}{\phi^L Q_L} +1 } 
  \geq \frac{L}{\frac{(L-1) \sup_{1\leq l \leq L} Q_l}{\phi Q_L} + 1} \to L \qquad\text{ as } \phi \to \infty \ . 
\end{align*}
But this is a contradiction, since $L\geq 1$ was arbitrary. So $\varrho_c<\infty$ implies $\phi_c<\infty$. The observation that $\varrho_c < \infty$ in addition implies $Z(\phi_c)<\infty$ is based on the identity
\[
  \log Z(\phi) = \log Z(0)+ \int_0^{\phi} \pderiv{}{\phi} \log Z(\phi) \dx{\phi} = \log Z(0) + \int_0^\phi \frac{\sum_{l\geq 1} l \, \phi^l Q_l}{\sum_{l\geq 0} \phi^l Q_l} \dx{\phi} \ ,
\]
which implies with $\varrho_c <\infty$ and $Z(0)=1$ the bound $Z(\phi) \leq e^{\varrho_c \phi}$. Togther with $\phi_c<\infty$ follows that 
\[
  Z(\phi_c) = \limsup_{\phi\uparrow \phi_c} Z(\phi) \leq e^{\phi_c\, \varrho_c } <\infty
\]
and hence $\omega^{\varrho_c}= \omega(\phi(\varrho_c))$ is well-defined. Hence, the set of all normalized equilibria is given by 
\begin{equation}\label{def:omega}
 \set[\big]{ \omega^\varrho = \omega(\phi(\varrho)) \ : \ \varrho<\infty, 0\leq \varrho \leq \varrho_c  } \ .
\end{equation}
The main tool of the proof of convergence to equilibrium is the free energy functional of the form
\begin{equation}\label{e:def:Lyapunov}
 \cF[c] = \sum_{k\geq 0} c_k \log \frac{c_k}{Q_k} \ ,
\end{equation}
which turns out to be a Lyapunov function for the evolution~\eqref{e:master} and the main tool in proving the following theorem.
\begin{thm}[Convergence to equilibrium]\label{thm:longtime}
  Suppose Assumption~\ref{ass:BDA} with $\phi_c\in (0,\infty)$~in~\eqref{e:ass:Kc} holds. Then for any $\varrho_0\in [0,\infty) $ and any $\bar c\in \cP^{\varrho_0}$ the unique solution~$c$ of~\eqref{e:master} with $c(0) = \bar c$ satisfies:
  \begin{enumerate}
   \item If $\varrho \leq \varrho_c$, it holds $c(t) \to \omega^{\varrho}$ strongly in $\cX$ as $t \to \infty$ and 
   \[
     \lim_{t\to \infty} \cF[c(t)] = \cF[\omega^{\varrho}] \ . 
   \]\item If $\varrho > \varrho_c$, it holds $c(t) \xrightharpoonup{*} \omega^{\varrho_c}$ as $t\to \infty$ and 
   \[ 
     \lim_{t \to \infty} \cF[c(t)] = \cF[\omega^{\varrho_c}] + \bra*{ \varrho - \varrho_c} \log \phi_c \ . 
   \]
  \end{enumerate}
\end{thm}
\begin{rem}
  Under the extra condition $\cF[\bar c]< \infty$ on the initial data, the strong convergence is also proven in the case $\phi_c=\infty$ in Corollary~\ref{cor:longtime:phc:infty}. Since $\cF$ is finite over $\cP^{\varrho}$ for $\phi_c\in (0,\infty)$, this condition is not needed in Theorem~\ref{thm:longtime}.
  
  With the characterization of weak$^*$ convergence in Proposition~\ref{prop:X:properties}, the statement in \emph{(2)} becomes just $c_k(t) \to \omega_k^{\varrho_c}$ for all $k\geq 0$. In particular, the excess mass $\varrho_0 - \varrho_c$ is lost in the limit $t\to \infty$. 
\end{rem}

\subsection{Formal gradient flow structure}\label{s:GF}

The free energy functional~$\cF$ from~\eqref{e:def:Lyapunov} is not only a Lyapunov functional for the system~\eqref{e:master}, but also the driving functional behind a gradient flow structure of the equation. This observation goes back to~\cite{Mielke2011a} for finite chemical reaction networks under detailed balance condition and to~\cite{Maas2011} in the setting of reversible Markov chains on finite state spaces. The key observation is that Assumption~\eqref{e:ass:BDA} or equivalently~\eqref{e:Q:DBC} is sufficient to define a suitable metric under which~\eqref{e:master} becomes the gradient flow of the free energy $\cF$.

The Assumption~\eqref{e:ass:BDA} makes the evolution to some extent symmetric, which can be seen by using~\eqref{e:Q:DBC} to define the symmetric quantity 
\[
  \kappa(k,l-1) = K(k,l-1) \, Q_k \, Q_{l-1} = K(l,k-1) \, Q_l \, Q_{k-1} = \kappa(l,k-1) \ .
\]
Therewith, the equation~\eqref{e:master:flux} can be rewritten as
\begin{equation}\label{e:Master:sym}
  \dot c = -\frac{1}{2} \sum_{k,l\geq 1} \kappa(k,l-1) \, \bra*{ \frac{c_k \, c_{l-1}}{Q_k \, Q_{l-1}} - \frac{c_{k-1} \, c_{l}}{Q_{k-1} \, Q_l }} \, \bra*{ \alpha^{k,l-1} - \alpha^{l,k-1}} \ ,
\end{equation}
where $\alpha^{k,l-1}$ are called stoichiometric coefficients and are given by
\[ 
  \alpha^{k,l-1}_r = \delta_{k,r} + \delta_{l-1,r}
  \qquad\text{with}\qquad 
  \delta_{k,r} = \begin{cases} 
                    1 &, k=r \\
                    0 &, k\ne r 
                  \end{cases} \ .
\]
The functional derivative of $\cF[c]$ is identified with 
\[
  D\cF[c] = \bra*{ \log \frac{c_k}{Q_k} - 1 }_{k\geq 0} \ .
\]
Then, the evolution~\eqref{e:Master:sym} can be written as the gradient flow 
\begin{equation}
  \dot c = - \cK[c] \; D \cF[c ]
\end{equation}
where the linear operator $\cK[c]$ is formally given by the infinite matrix
\begin{equation}
  \cK[c] = \frac{1}{2}\sum_{k,l\geq 1}  \kappa(k,l-1) \; \Lambda_{\Boltz}\bra*{\frac{c_k \, c_{l-1}}{Q_k \, Q_{l-1}} , \frac{c_{k-1} \, c_{l}}{Q_{k-1} \, Q_l }} \  \bra*{ \alpha^{k,l-1} - \alpha^{l,k-1}}  \otimes  \bra*{ \alpha^{k,l-1} - \alpha^{l,k-1}} \ .
\end{equation}
Hereby, $\Lambda: \R_{\geq 0} \times \R_{\geq 0} \to \R_{\geq 0}$ is the logarithmic mean
\begin{equation}
  \Lambda_{\Boltz}(s,t) = \begin{cases} 
                            \frac{s-t}{\log s - \log t }\ , & s\ne t \\
                            s \ , & s=t 
                          \end{cases}\ .
\end{equation}

\subsection{Open questions}

The Assumption~\ref{ass:well-posedness} leaves parameter regimes open. There is a minimal gap between assumption~\eqref{e:ass:K1} for existence and the slightly stronger one~\eqref{e:ass:K2} for uniqueness. Note, that by the results of~\cite[Theorem 5]{emre} follows uniqueness under assumption~\eqref{e:ass:K1} for initial values with finite $p$-moment for any $p>1$, that is $\sum_{k\geq 1} k^p \bar c_k <\infty$. 

The picture for the longtime behavior is a bit less complete due to the various conditions in Assumption~\ref{ass:BDA}. First of all~\eqref{e:ass:K4} implies in particular sublinear growth for the kernel, and hence the case of kernels with linear growing rates is open. This assumption enters the proof at two points. Firstly, to establish that the solutions to~\eqref{e:master} constitute a semigroup in the weak$^*$ topology (Theorem~\ref{thm:generalized_flow:weak}) the sublinear growth of $\set{a_l}_{l\geq 0}$ and $\set{b_k}_{k\geq 1}$ is needed. Although, it seems possible to salvage this by using similar arguments as in~\cite{Slemrod1989}. 
Secondly, the relative compactness of the orbits for solutions with initial mass density $\varrho<\varrho_c$ is only established under~\eqref{e:ass:K4} with $\set{a_l}_{l\geq 0}$ and $\set{b_k}_{k\geq 1}$ with at most linear growth. 

For kernels not satisfying the detailed balance assumption~\eqref{e:ass:BDA}, the longtime behavior is largely open. The recent preprint~\cite{EsenturkVelazquez2019} provides a contractivity statement for any two solutions under suitable monotonicity assumptions on the rates, which in particular gives convergence to equilibrium. 

Also, once a qualitative convergence statement as \emph{(1)} in Theorem~\ref{thm:longtime} is proven, one could ask for an improvement to a quantitative statement, which seems possible by the tools developed in~\cite{Canizo2015,Conlon2017} under suitable additional assumptions on the kernel.

The statement \emph{(2)} of Theorem~\ref{thm:longtime} raises the question of how the escape of the excess mass $\varrho_0-\varrho_c$ happens and if some evolution equation may be deduced, which is asymptotically satisfied by the excess mass. The similarity to the Becker-Döring model suggests that some transport equation related to the classical theory for coarsening by Lifshitz–Slyozov~\cite{Lifshitz1961} and Wagner~\cite{Wagner1961} may be deduced by similar means as in~\cite{Penrose1997,Velazquez1998,LM02,Niethammer2003,Schlichting2018}.

The formal framework in section~\ref{s:GF} could be made rigorous by following the approach of~\cite{EFLS16}. In addition to that, the mean-field limit of the stochastic particle systems as obtained by~\cite{JG2017} could be proven within the variational framework of (generalized) gradient flows or similarly in the context of large-deviations. A related question is whether the variational evolutionary $\Gamma$-convergence as applied in~\cite{Schlichting2018} to the Becker-Döring system is applicable to obtain a macroscopic limit of the exchange-driven growth model.

\subsection*{Outline}

The section~\ref{s:well-posedness} contains the well-posedness results, that is Theorem~\ref{thm:well-posed}. First, existence via a truncation argument is shown in subsection~\ref{s:exist}. After this, the fact that solutions constitute a generalized flow is obtained in section~\ref{s:flow}. Finally, uniqueness with the semigroup property and positivity is proven in section~\ref{s:unique}. The longtime behavior as stated in Theorem~\ref{thm:longtime} is proven in section~\ref{s:equilibrium}. First, the Lyapunov property of~\eqref{e:def:Lyapunov} is established in subsection~\ref{s:lyap}. Secondly, the relative compactness of orbits is obtained in subsection~\ref{s:compact}, which allows to apply LaSalle's invariance principle later. The proof of Theorem~\ref{s:longtime} concludes in subsection~\ref{s:longtime}. The appendix~\ref{s:ValleePoussin} provides a version of the Lemma de la Vallée-Poussin, which is needed for existence.
\section{Well-posedness}\label{s:well-posedness}
\subsection{Existence by truncation}\label{s:exist}
The basic properties of the space $\cX$ from~\eqref{e:def:X} are summarized below.
\begin{prop}[{\cite{BCP86}}]\label{prop:X:properties}
  The space $\cX$ is a Banach space and it is the dual space of
  \[
    ^*\cX= \set*{ (c_l)_{l\geq 0} \ :\ (1+l)^{-1} c_l \to 0 \text{ as } l \to \infty } \ . 
  \]
  Moreover, let a sequence $(c^{j})_{j\geq 0} \subset \cX$ and some $c \in \cX$ be given. Then
  \begin{enumerate}
   \item $c^{j}$ converges weakly$^*$ to $c$ in $\cX$ if and only if
   \begin{enumerate}
    \item $\sup_j \norm{ c^{j}} < \infty$, and
    \item $c^{j}_l \to c_l$ as $j\to \infty$ for all $l\geq 0$.
    \end{enumerate}
   \item $c^{j}$ converges strongly to $c$ in $\cX$ if and only if
    \begin{enumerate}
    \item $c^{j}$ converges weakly$^*$ to $c$ in $\cX$, and
    \item $\norm{ c^{j} } \to \norm{c}$ as $j \to \infty$. 
    \end{enumerate}
  \end{enumerate}
  For a given sequence $\set{\alpha_k}_{k\in \N_0}$, the functional $\cA: \cX \to \R$ given by
  \begin{equation}\label{e:potential:energy}
    \cA[c] = \sum_{k\geq 0} \alpha_k \, c_k 
  \end{equation}
  is weak$^*$ continuous if and only if $\alpha_k / k \to 0$ as $k\to \infty$. 
\end{prop}
Any sequence in $\cB^\varrho$ from~\eqref{e:def:ball} is by definition bounded in $\cX^+$ and has by Proposition~\ref{prop:X:properties} a weak$^*$ convergent subsequence. The topology of weak$^*$ convergence in~$\cB^\varrho$ can be metrizised by
\begin{equation}\label{e:def:upd}
  \upd(c,d) = \sum_{k\geq 0} \abs*{ c_k -d_k} \ .
\end{equation}
The space $\cB^\varrho$ becomes a compact metric space equipped with $\upd$ and to avoid any confusion this space is denoted by $\cB^\varrho_\upd$.
\begin{defn}[Solution]\label{def:sol}
  For $T\in [0,\infty]$ a family of functions $\set{c_k(\cdot)}_{k\geq 0}$ is called a solution to~\eqref{e:master} on $[0,T)$ provided that
  \begin{enumerate}
   \item $c_k: [0,T ) \to [0,\infty)$ is continuous and bounded $\sup_{t\in [0,T)} c_k(t) < \infty$
   \item The nonlinear birth $A_{k-1}[c]$ and death rates $B_k[c]$ defined in~\eqref{e:bdrates} are integrable
   \begin{equation}\label{e:def:sol:bounded:birth-death}
     \int_0^t A_{k-1}[c(s)] \dx{s} < \infty \quad\text{and}\quad \int_0^t B_k[c(s)] \dx{s} < \infty \quad \text{for $k\geq 1$ and } t\in [0,T) \ . 
   \end{equation}
   \item The equation~\eqref{e:master:flux} holds in integrated form for $t\in [0,T)$
   \begin{equation}\label{e:def:sol:flux}
     c_k(t) = c_k(0) + \int_0^t \bra[\big]{ J_{k-1}[c(s)]  - J_{k}[c(s)]} \dx{s} \qquad \text{ for $k\geq 0$}\ ,
   \end{equation}
   again with the convention that $J_{-1}\equiv 0$.
  \end{enumerate}
\end{defn}
Definition~\ref{def:sol} is also used in~\cite{emre}, where the conservation laws~\eqref{e:ConverationLaws}, positivity, existence and uniqueness were deduced under additional assumptions on initial moments. The result proven in this section extends the well-posedness theory of~\cite{emre} to arbitrary initial data under the sole Assumption~\ref{ass:well-posedness}. The first step for the existence and stability of solution is done by considering for $N\geq 1$ the following truncated system of ordinary differential equations.
\begin{align}\label{e:master:trunc}
  \dot c_k^{N} = J^{N}_{k-1}[c^{N}] - J^{N}_k[c^{N}] \qquad\text{ for $k=0,\dots, N$ } \ ,
\end{align}
where
\begin{equation}\label{e:net:flux:trunc}
 J^{N}_k[c^{N}] = A_k^{N}[c^{N}] \, c_k^N - B_{k+1}^N[c^N] \, c_{k+1}^N \qquad\text{ for $k=0,\dots N-1$} \ .
\end{equation}
with
\begin{equation}
  A^N_{k-1}[c^N] = \sum_{l=1}^N K(l,k-1) \, c_l^N \quad\text{and}\quad  B_k^N[c^N] = \sum_{l=0}^{N-1} K(k,l) \, c_l^N \qquad\text{ for $k=1,\dots, N$} \ .  \label{e:def:AB:trunc} 
\end{equation}
Any element $c^{N} \in \R^{N+1}$ is extended to $\cX$ by setting $c_k^{N} = 0 $ for $k > N$. 

The well-posedness of the truncated system~\eqref{e:master:trunc} follows from standard theory of ordinary differential equations. However, to deduce stability properties of the infinite system~\eqref{e:master}, certain estimates for~\eqref{e:master:trunc}, uniform in $N$, are needed. First properties of the truncated system are deduced by the following simple Proposition, which is also the basis of the analysis in~\cite{emre}. It is convenient to rewrite~\eqref{e:master:trunc} as a nonlinear birth-death chain based on the above definitions 
\begin{align}
  \dot c_0^{N} &= - A_0^{N}[c^N] \, c_0^N + B_1^N[c^N] \, c_1^N \ ; \notag \\
  \dot c_k^{N} &= A_{k-1}^N[c^{N}] \, c_{k-1}^{N} - \bra*{ A_k^N[c^{N}] + B_k^N[c^{N}] } \, c_k^N  + B_{k+1}^N[c^{N}] \, c_{k+1}^{N}  \qquad \text{ for } k=1,\dots, N-1 \ ;\notag \\
  \dot c_N^{N} &= A_{N-1}^N[c^N] \, c_{N-1}^N - B_N^N[c^N] \, c_{N}^N \label{e:master:trunc:birth-death} \ .   
\end{align}
The basic properties of the truncated system are already established in~\cite{emre}.
\begin{prop}[Properties of truncated system {\cite[Lemma 1+2, Corollary 1]{emre}}]\label{prop:trunc:system}
  For any $N$ let $c^N$be the solution to~\eqref{e:master:trunc}. Then it holds for any sequence of real numbers~$(g_k)_{k \geq 0}$ 
  \begin{equation}\label{e:trunc:test}
  \pderiv{}{t} \sum_{k=0}^N \, g_k \, c_k^N + \sum_{k=1}^N \bra*{ g_k - g_{k-1}} \, B_k^N[c^N] \, c_k^N 
  = \sum_{k=0}^{N-1} \bra*{ g_{k+1}-g_k} \,  A_k^N[c^N] \, c_k^N   \ .
  \end{equation}
  In addition, the zeroth and first moments of $c^N$ are conserved. Moreover, nonnegativity of the initial data $c^N_k(0) \geq 0$ for $k\geq 0$ is preserved $c_k^N(t)\geq 0$ for all $t\in [0,\infty)$. In particular, for $c^N(0) \in  \cP^\varrho$ follows $c^N(t) \in \cP^\varrho$ for all $t\in [0,\infty)$.
\end{prop}
Based on the above Proposition, the existence of solutions is obtained by a suitable limiting procedure with the help of the Arzelà-Ascoli Theorem similar to how it is done in~\cite[Theorem 2.2]{BCP86} for the classical Becker-Döring system.
\begin{thm}[Existence of solutions]\label{thm:stability}
  Let the rates satisfy the linear growth assumption for all $l,k\geq 1$
  \begin{equation}\label{e:linear:growth2}
    K(k,l-1) \leq a_k \, a_l  \ ,  \qquad\text{where for some $C_K>0$: }\quad  a_k \leq C_K \, k \ . 
  \end{equation}
  Let $(g_k)_{k\geq 0}$ be a positive increasing sequence satisfying for some $C_g>0$
  \begin{equation}\label{e:ass:a:g}
    a_k \, \bra*{ g_{k+1} - g_{k}} \leq C_g \, g_k \ . 
  \end{equation}
  If $a_k$ has linear growth, that is $\liminf_{k\to \infty}\frac{a_k}{k} >0$, then $g_k$ is additionally assumed to be of superlinear growth, that is $\lim_{k\to \infty} \frac{g_k}{k} =\infty$. 
  
  Then for any $T>0$, any $\varrho>0$ and any $\bar c \in \cP^\varrho$ with $\sum_{k\geq 0} g_k \, \bar c_k \leq \bar C_g < \infty$ there exists a nonnegative solution to~\eqref{e:master} with $c(0)=\bar c$ satisfying the bound
  \begin{equation}\label{e:stability}
     \sum_{k\geq 0} g_k \, c_k(T) + \int_0^T \sum_{k\geq 1} \bra*{g_k-g_{k-1}} \, B_k[c(s)] \, c_k(s) \dx{s} \leq\bar C_g e^{C_K\, (1+\varrho)\, C_g \, T} \ . 
  \end{equation}
  Moreover, for any $\bar c\in \cP^\varrho$ there exists a superlinear positive increasing sequence $(g_k)_{k\geq 0}$ satisfying~\eqref{e:ass:a:g} with $\sum_{k\geq 0} g_k \bar c_k < \infty$. In particular, there exists a nonnegative solution to~\eqref{e:master} with $c(0)=\bar c$ conserving the number density and total mass
  \begin{equation}\label{e:trunc:convervation}
    \sum_{k\geq 0} c_k(t) = \sum_{k\geq 0} \bar c_k =1  \qquad\text{and}\qquad \sum_{k\geq 1} k \, c_k(t) = \sum_{k\geq 1} k\, c_k(t) = \varrho \qquad\text{for all } t\in [0,T) \ . 
  \end{equation}
\end{thm}
\begin{rem}
  Theorem~\ref{thm:stability} also contains the stability of solutions on compact time intervals. For this reason, the quantified growth condition~\eqref{e:linear:growth2} is introduced. Especially, by choosing $a_k = k+1$ and $g_k=(k+1)^p$ for $p\geq 1$, for which~\eqref{e:ass:a:g} holds with $C_g=2^p-1$, it shows that arbitrary high moments are bounded on compact time intervals, once the initial data has a $p$-th moment. 
  
  Likewise, Theorem~\ref{thm:stability} contains the existence part of Theorem~\ref{thm:well-posed}, for which it is applied with $g_k=k+1$ in the case where $a_k$ is of sublinear growth. In this case, \eqref{e:ass:a:g} is satisfied with $C_g=1$ and also the condition on the initial datum is satisfied, since $\bar c \in \cP^\varrho$ and therefore $\sum_{k\geq 0} (k+1)\, \bar c_k = 1 + \varrho = \bar C_g < \infty$. If $a_k$ is of linear growth, then a superlinear sequence $g_k$ satisfying~\eqref{e:ass:a:g} and $\sum_{k\geq 0} g_k\, \bar c_k <\infty$ is obtained from $\bar c\in \cP^\varrho$ via the Lemma de la Vallée-Poussin~\ref{lem:ValleePoussin}.
\end{rem}
\begin{proof}
Let $\bar c^N_k = \bar c_k$ for $k=0,\dots, N$. Then Proposition~\ref{prop:trunc:system} implies the existence of a unique solution $c^N$ to~\eqref{e:master:trunc} with $c^N(0) = (\bar c_k)_{0\leq k\leq N}$ and $c^N(t)$ preserves for $t\in [0,\infty)$ the zeroth and first moment, which in particular implies the bounds
\begin{equation}
  \sum_{l=0}^N c_l^N(t) = \sum_{l=0}^N \bar c_l  \leq 1 \qquad\text{and}\qquad \sum_{l=1}^N l \, c_l^N(t) = \sum_{l=0}^N l\, \bar c_l \leq \varrho \qquad\text{for all } t\in [0,\infty) \ . 
\end{equation}
The last bound translates to the pointwise bounds
\begin{align}\label{e:cN:bd}
  0\leq c_0^N(t) \leq 1 \qquad\text{and}\qquad 0 \leq  c_k^N(t) \leq \frac{\varrho}{k}  \qquad\text{for $k=1,\dots, N$} \ . 
\end{align}
The at most linear growth assumption~\eqref{e:linear:growth2} implies for all $k\geq 1$
\begin{align}
  A^N_{k-1}[c^N] &= \sum_{l=1}^N K(l,k-1) \, c_l^N \leq C_K^2 \, \varrho \, k \label{e:AN:bd} \\
  B^N_{k}[c^N] &= \sum_{l=1}^{N-1} K(k,l-1) \, c_{l-1}^N \leq C_K^2 \, (\varrho + 1) \, k \label{e:BN:bd}  \ .
\end{align}
Plugging these bounds into~\eqref{e:master:trunc:birth-death} yields for all $k= 0,\dots, N$ the estimate
\begin{align}
  \abs{\dot c_k^N} &\leq A_{k-1}^N[c^N] c_{k-1}^N  + \bra*{A_k^N[c^N]+ B_k^N[c^N]} c_{k}^N + B_{k+1}^N[c^N] c_{k+1}^N \notag\\
  &\leq C_K^2 \bra*{ \varrho (k-1) c_{k-1}^N + (2\varrho+1) k c_k^N + \bra*{\varrho+1} (k+1) c_{k+1}^N } \notag \\ 
  &\leq 2\,C_K^2 \, (2\varrho +1) \sum_{k= 1}^N k c_{k}^N = 2\,C_K^2 \, (2\varrho+1) \, \varrho  < \infty \ . \label{e:dotcN:bd}
\end{align}
This shows that, for each $k\geq 0$, the family $\bra*{c_k^N(\cdot)}_{N\geq k}$ is equicontinuous on $[0,\infty)$ and the Arzelà-Ascoli theorem implies after extracting a suitable diagonal subsequence $N_n \to \infty$, the existence of a sequence of continuous function $c_k : [0,\infty) \to \R$ with $c_k^{N_n} \to c_k$ uniformly on compact subintervals of $[0,\infty)$ for all $k\geq 0$. In particular, the limit satisfies 
\begin{equation}\label{e:exits:p0}
  \sum_{l\geq 0} c_l(t) \leq 1 \qquad\text{and}\qquad \sum_{l\geq 1} l c_l(t)\leq \varrho \qquad\text{for } t\in [0,\infty) \ .
\end{equation}
The uniform bounds~\eqref{e:AN:bd} and~\eqref{e:BN:bd} together with the uniform convergence give for any $M\geq 1$ and $k\geq 1$ the bounds
\begin{align*}
  \sum_{l=1}^M K(l,k-1) \, c_l \leq C_K^2 \, \varrho \, k  \qquad\text{ and }\qquad \sum_{l=1}^{M-1} K(k,l-1) \, c_{l-1}^N \leq C_K^2 \, (\varrho + 1) \, k  \ .
\end{align*}
Integration of these bounds over $[0,t)$ and letting $M\to \infty$ gives~\eqref{e:def:sol:bounded:birth-death}. 

Before turning to the proof of~\eqref{e:def:sol:flux}, the stability estimate~\eqref{e:stability} based on~\eqref{e:trunc:test} from Proposition~\ref{prop:trunc:system} is shown. The two assumptions~\eqref{e:linear:growth2} and~\eqref{e:ass:a:g} allow to bound the right hand side of~\eqref{e:trunc:test} 
\[
  \sum_{k=0}^{N-1} \bra*{ g_{k+1}-g_k} \, A_k^N[c^N] \, c_k^N \leq \sum_{k=0}^{N-1} \bra*{g_{k+1}-g_k} \, a_k \, c_k^N  \sum_{l\geq 1} a_l \, c_l^N  \leq C_K\, (1+\varrho) \, C_g \sum_{k=0}^{N-1}  g_k \, c_k^N  \ .
\]
Plugging this bound into~\eqref{e:trunc:test} yields
\[
  \pderiv{}{t} \sum_{k=0}^N g_k \, c_k^N + \sum_{k=1}^N \bra*{ g_k - g_{k-1}} \, B_k^N[c^N] \, c_k^N 
  \leq C_K\, (1+\varrho)\, C_g \sum_{k=0}^{N}  g_k \, c_k^N \ . 
\]
The Gronwall Lemma gives for any $T\in [0,\infty)$ the uniform in time bound 
\begin{equation}
  \sum_{k=0}^N g_k \, c_k^N(T) + \int_0^T \sum_{k=1}^N \bra*{ g_k - g_{k-1}} \, B_k^N[c^N(t)] \, c_k^N(t) \dx{t} \leq \bar C_g e^{ C_K\, (1+\varrho)\, C_g T} \ . 
\end{equation}
The map $k\mapsto g_k$ is increasing, which allows to pass to the limit $N\to \infty$ along the above used subsequence $N_n$ leading to the bound~\eqref{e:stability}.

To pass to the limit in the integrated form of~\eqref{e:master:trunc} to obtain~\eqref{e:def:sol:flux}, it is needed to show
\begin{equation}\label{e:exists:p1}
  \lim_{N\to \infty} \int_0^T A_{k-1}^N[c^N(s)] \dx{s} = \int_0^T A_{k-1}[c(s)]\dx{s} \ .
\end{equation}
By the uniform convergence of each $c_k^N(\cdot) \to c_k(\cdot)$ on compact time intervals follows
\begin{equation}
  \sum_{l=1}^M K(l,k-1) c_l^N(t) \to \sum_{l=1}^M K(l,k-1) c_l(t) \qquad\text{ for any } t\in [0,\infty) \text{ as } N\to \infty \ . 
\end{equation}
If $(a_k)_{k\geq 0}$ is of sublinear growth, then the bound~\eqref{e:exits:p0} allows to let $M\to \infty$. If $(a_k)_{k\geq 0}$ is of linear growth, then the superlinear growth of $(g_k)_{k\geq 0}$ in the just established stability estimate~\eqref{e:stability} gives 
$\sum_{k\geq 0} g_k \, c_k(t) < \infty$ for any $t\in[0,\infty)$, which allows again to let $M\to \infty$ also in this case. Hence, the limit \eqref{e:exists:p1} is shown and a similar argument allows to pass to the limit in $B_k^N[c(\cdot)]$. Both bounds allow to pass to the limit in the integrated form of~\eqref{e:master:trunc} to obtain~\eqref{e:def:sol:flux}. 
Hence, the family $\bra{c_k(\cdot)}_{k\geq 0}$ is a solution to~\eqref{e:master} in the sense of Definition~\ref{def:sol}.

Finally, the existence of a superlinear positive increasing sequence $(g_k)_{k\geq 0}$ satisfying~\eqref{e:ass:a:g} with $\sum_{k\geq 0} g_k \bar c_k < \infty$ is established in the Lemma de la Vallée-Poussin~\ref{lem:ValleePoussin}. The stability bound~\eqref{e:stability} allows then to pass to the limit in the conservation of the truncated system established in Proposition~\ref{prop:trunc:system}.
\end{proof}
\subsection{Generalized flow}\label{s:flow}
After having established the existence of solution, the next basic property, which should be satisfied by all solutions in the sense of Definition~\ref{def:sol} are the conservation laws~\eqref{e:ConverationLaws}, such that every solution is actually in $\cP^\varrho$ on compact time intervals. This is a consequence of a theorem from~\cite{BCP86}, which carries over with only minor modifications and its proof is omitted.
\begin{thm}[{\cite[Theorem 2.5]{BCP86}}]\label{thm:weak_form}
 Let $(g_k)_{k\geq 0}$ be a real sequence. Let $c$ be a solution to~\eqref{e:master} on some interval $[0,T)$ with $0< T \leq \infty$ in the sense of Definition~\ref{def:sol}. Suppose that for $0\leq t_1 < t_2 < T$, $\int_{t_1}^{t_2} \sum_{k\geq 0} \abs{g_{k+1} - g_k} A_k[c(t)] c_k(t) \dx{t} < \infty$ and either that $g_k = O(k)$ and $\int_{t_1}^{t_2} \sum_{k\geq 0} \abs{g_{k+1} - g_k} B_{k+1}[c(t)] c_{k+1}(t) \dx{t} < \infty$ or that $\sum_{k\geq 0} g_k c_k(t_i) < \infty$ for $i=1,2$ and $g_{k+1} \geq g_k \geq 0$ for sufficiently large $k$. Then, for all $m\geq 0$
 \begin{align}
   \sum_{k\geq m} g_k \, c_k(t_2) &- \sum_{k\geq m} g_k \, c_k(t_1) + \int_{t_1}^{t_2} \sum_{k\geq m} \bra*{g_{k+1} - g_k} \, B_{k+1}[c(t)] \, c_{k+1}(t) \dx{t} \label{e:flux:identity} \\
   &= \int_{t_1}^{t_2} \sum_{k\geq m} \bra*{ g_{k+1} - g_{k}} \, A_{k}[c(t)] \, c_k(t) \dx{t} + \int_{t_1}^{t_2} g_m \, J_{m-1}[c(t)] \dx{t} \ .\nonumber 
 \end{align}
\end{thm}
The conservation of mass is a direct consequence of the above statement.
\begin{cor}[Conservation laws]\label{cor:convervation_laws}
  Let $c$ be a solution to~\eqref{e:master} with $c(0) = \bar c\in \cP^\varrho$ on some interval $[0,T)$ for $0< T \leq \infty$. Then for all $t\in [0, T)$ holds $c(t) \in \cP^\varrho$. 
\end{cor}
\begin{proof}
  Setting $m=0$ and $g_k = 1$ for all $k$ in~\eqref{e:flux:identity} yields $\sum_{k\geq 0} c_k(t)= \sum_{k\geq 0} c_k(0)= 1$. Similarly, choosing $m=0$ and $g_k = k$ gives for any $0\leq t_1 \leq t_2 < T$ the identity 
  \[
     \sum_{k\geq 1} k \, c_k(t_2) - \sum_{k\geq 0} k \, c_k(t_1) + \int_{t_1}^{t_2} \sum_{k\geq 0} B_{k+1}[c(t)] \, c_{k+1}(t) \dx{t} 
   = \int_{t_1}^{t_2} \sum_{k\geq 0} A_{k}[c(t)] \, c_k(t) \dx{t} \ .
  \]
  The conservation of the first moment follows from noting that
  \[
    \sum_{k\geq 1} B_{k}[c] \, c_{k} = \sum_{k\geq 1} \sum_{l\geq 1} K(k,l-1) \, c_{k} \,c_{l-1} = \sum_{k\geq 1} A_{k-1}[c] \, c_{k-1} \ .
  \]
\end{proof}
Another consequence of the mass conservation is the continuity of solutions, which proof follows along the lines of~\cite[Proposition 3.1]{BCP86}. 
\begin{prop}\label{prop:cont:solution}
  Let $c$ be a solution to~\eqref{e:master} on some interval $[0,T)$ for $0< T \leq \infty$. Then $c: [0,T) \to \cX$ is continuous and the series $\sum_{k\geq 0} (1+k) \, c_k(t)$ is uniformly convergent on compact intervals of $[0,T)$. 
\end{prop}
\begin{proof}
  For $l\geq 0$, the functions $f_l(t) = \sum_{k=0}^l (1+k)\, c_k(t)$ are continuous and monotone $f_{l+1}\geq f_l$ on $[0,T)$. By the mass conservation from Corollary~\ref{cor:convervation_laws}, it follows $\lim_{l\to \infty} f_l(t) = 1+ \varrho$. Hence, the uniform convergence statement is a consequence of Dini's theorem and the continuity of $c$ in $\cX$ is a consequence of the continuity of the individual $c_k$ for all $k\geq 0$.
\end{proof}
\begin{defn}[Semigroup/Generalized flow]\label{def:flows}
  A \emph{generalized flow} $G$ on a metric space $Y$ is a family of continuous mappings $\varphi: [0,\infty) \to Y$ such that
 \begin{enumerate}
   \item If $\varphi(\cdot)\in G$, so is for any $\tau>0$ also $\varphi(\cdot +\tau)\in G$.
   \item For all $y\in Y$ exists $\varphi(\cdot)\in G$ with $\varphi(0)= y$.
   \item If a family $\bra{\varphi^{j}}_{j\geq 0}\subset G$ satisfies $\varphi^{j}(0) \to \varphi(0)$ in $Y$, then there exists a subsequence $(j_k)_{k\geq 0}$ and $\varphi \in G$ such that $\varphi^{j_k}(t) \to \varphi(t)$ on compact time intervals.
  \end{enumerate}
  If $G$ is a generalized flow such that for each $y\in Y$ exists a \emph{unique} $\varphi \in G$ with $\varphi(0)=y$, then $G$ is called~\emph{semigroup}. In this case, for any $t\geq 0$ the mapping $T(t): Y\to Y$ defined by $T(t)y= \varphi(t)$ satisfies
  \begin{enumerate}
   \item $T(0) = \operatorname{Id}$ ,
   \item $T(t+s) = T(s) T(t)$ for all $t,s\geq 0$ ,
   \item $(t,y) \mapsto T(t) y$ is continuous from $[0,\infty) \times Y \to Y$.
  \end{enumerate}
\end{defn}
The fact that all solutions constitute a generalized flow with respect to the strong topology is shown along the lines of~\cite[Theorem 3.4]{BCP86}.
\begin{thm}[Generalized flow in strong topology]\label{thm:generalized_flow}
  Let $G$ the set of all solutions $c$ to~\eqref{e:master} with $c(0) \in \cP = \set{ c \in \cX^+ : \sum_{l\geq 0} c_l = 1}$. Then $G$ is a generalized flow on the subspace $\cP$ of $\cX$. 
\end{thm}
\begin{proof}
  Any solution $c:[0,\infty) \to  \cP$ to~\eqref{e:master} is continuous by Proposition~\ref{prop:cont:solution}. The existence of a solution to initial data $\bar c \in \cP$ is a consequence of Theorem~\ref{thm:stability}.
  This shows the first two properties of a generalized flow in Definition~\ref{def:flows}. It is left to prove the third semicontinuity property. Hence, let $\bar c^{j}$ be a sequence of initial data converging to $\bar c$ in $\cX$. Now, let $c^{j}$ be a sequence of solutions with $c^{j}(0) = \bar c^{j}$. Since, $\bar c^j \to \bar c$ in $\cX$, it follows that $\cC = \set{\bar c^j}_{j\geq 0} \cup \bar c$ is tight in $\cX$ and there exists a superlinear function $(g_k)_{k\geq 0}$ satisfying~\eqref{e:ass:a:g} such that $\sup_{j} \sum_{k\geq 0} g_k \bar c_k^j < \infty$ based on the Lemma of de la Vallée-Poussin~\ref{lem:ValleePoussin}. Hence, the family $(c^j(\cdot))_{j}$ constructed in Theorem~\ref{thm:well-posed} is also tight satisfying the bound $\sup_{t\in [0,T)} \sup_j \sum_{k\geq 0} g_k c_k^j(t) <\infty$ for any $T>0$ by the stability estimate~\eqref{e:stability}.  Hence there is $c\in C([0,\infty),\cX^+)$ such that, after the possible extraction of a subsequence $(j_k)_{k \geq 0}$, $c^{j_k}_l(t) \to c_l(t)$ uniformly on $[0,T]$ for any $T>0$ and all $l\geq 0$. Then, by similar argument as in the proof of Theorem~\ref{thm:stability}, $c$ is a solution to~\eqref{e:master}. Moreover, the tightness of the family of solutions implies the conservation of the zeroth and first moment
  \[
    \sum_{l\geq 0} (l+1) \, c_l^{j_k}(t) = \sum_{l\geq 0} (l+1) \, c_l^{j_k}(0) \to \sum_{l\geq 0} (l+1) \, c_l(0) = \sum_{l\geq 0} (l+1) \, c_l(t) \quad\text{for $k\to \infty$}\ .
  \]
  In particular, convergence in $\cX$ on compact time intervals holds by Proposition~\ref{prop:X:properties}.
\end{proof}
Likewise all solutions constitute a generalized flow with respect to the weak$^*$ topology under the additional assumption of sublinear growth of the kernel, which is a result analog to~\cite[Theorem 3.5]{BCP86}.
\begin{thm}[Generalized flow in weak topology]\label{thm:generalized_flow:weak}
  Let the rates satisfy the sublinear growth assumption
  \begin{equation}\label{e:ass:sublinear}
    K(k,l-1) \leq a_k \, a_l  \ ,  \qquad\text{where }\quad \lim_{k\to \infty} \frac{a_k}{k} = 0 \ , \tag{$\mathrm{K}_1'$}
  \end{equation}
  and let $G$ the set of all solutions $c$ to~\eqref{e:master} with $c(0) \in \cX^+$.  Then for any $\varrho>0$ is $G$ a generalized flow on the compact metric space $(\cB^\varrho,\upd)$ with $\upd$ defined in~\eqref{e:def:upd}.
\end{thm}
\begin{proof}
 Thanks to the existence from Theorem~\ref{thm:stability} it is left to prove (3) of Definition~\ref{def:flows}. Hence, let $\bar c^j \weakto \bar c$ as $j\to \infty$ and let $(c^j)_{g\geq 1}$ be the sequence of solutions to~\eqref{e:master} constructed in Theorem~\ref{thm:stability}. The solutions are uniformly absolutely continuous satisfying $\dot c_l^j(t) \leq C_K (\varrho+1)\varrho$ for all $l\geq 0$ by the bound~\eqref{e:dotcN:bd}. Hence, by the Arzelà-Ascoli theorem, there exists a diagonal sequence $(j_k)_{k\geq 1}$ such that $c^{j_k}_l \to c^{j_k}_l$ uniformly in $l\geq 0$ on compact time intervals, which implies that $\upd\bra*{c^{j_k}(t),c(t)}\to 0$ on compact time intervals. Also, this convergence implies that
 \[
   \sum_{l\geq 1} l c_l(t) \leq \liminf_{k\to\infty} \sum_{l\geq 1} l c_l^{j_k}(t) = \varrho 
 \]
 for all $t\geq 0$, so that $c(t)\in \cB^\varrho$. Finally, it is possible to pass to the limit in the equation, because the coefficients $A_k[\cdot]$ and $B_k[\cdot]$ are of the form \eqref{e:potential:energy} and hence weak$^*$ continuous thanks to the sublinear growth assumption~\eqref{e:linear:growth2} by Proposition~\ref{prop:X:properties}. 
\end{proof}
\subsection{Uniqueness, semigroup and positivity}\label{s:unique}
The uniqueness result is based on ideas from~\cite{LM02}. 
It requires to slightly enforce the Assumption~\eqref{e:linear:growth2} by additionally requiring some regularity on the exchange rates.
\begin{thm}[Uniqueness]\label{thm:unique}
  If $K$ satisfies Assumption~\eqref{e:ass:K2}, then for any $\bar c\in \cP^\varrho$ with $\varrho>0$ and all $T>0$ exists a unique solution $c$ of~\eqref{e:master} on $[0,T)$ satisfying $c(0) = \bar c$.
\end{thm}
\begin{proof}
  Let $c$ be the solution to $\bar c$ constructed in Theorem~\ref{thm:stability} and $d$ another solution with the same initial datum $\bar c$. The core idea from~\cite{LM02} is to consider the tail distributions
  \[
    C_j(t) = \sum_{k\geq j} c_k(t) \qquad \text{and}\qquad D_j = \sum_{k\geq j} d_j(t) \ . 
  \]
  Proposition~\ref{prop:cont:solution} implies that $C \in C([0,T); \ell^1(\N_0))$, since 
  \[
    \sum_{j\geq 0} C_j = \sum_{k\geq 0} c_k \sum_{j=0}^k = \sum_{k\geq 0} (k+1)\, c_k = 1 + \varrho \ . 
  \]
  Furthermore, it holds $C_0(t) = 1 = D_0(t)$ for all $t \in [0,T)$. The differences
  \[
    E_k(t) = C_k(t) - D_k(t) = \sum_{j\geq k} e_j(t) \qquad\text{with}\qquad e_j(t) = c_j(t) - d_j(t) 
  \]
  satisfy by Theorem~\ref{thm:weak_form} applied with $g_k = 1$
  \begin{align*}
    \pderiv{E_k(t)}{t} &= J_{k-1}[c] - J_{k-1}[d]\\
    &=  A_{k-1}[e] \, c_{k-1} - B_{k}[e] \, c_{k}  + A_{k-1}[d] \, \bra*{E_{k-1}- E_k}  - B_{k}[d] \, \bra*{E_k - E_{k+1}}  \ .
  \end{align*}
  For any absolutely continuous function $\sigma : [0,T) \to \R$ holds by the chain rule $\pderiv{}{t} \abs{\sigma(t)} = \operatorname{sgn} \sigma(t) \; \dot \sigma(t)$ for a.e. $t\in [0,T)$. Hence, carefully tracking the signs results in the estimate
  \[
    \pderiv{\abs{E_k(t)}}{t} \leq \abs[\big]{A_{k-1}[e]} \, c_{k-1} + \abs[\big]{B_{k}[e]} \, c_{k} + A_{k-1}[d] \, \bra[\big]{\abs{E_{k-1}} - \abs*{E_k}}  + B_{k}[d] \, \bra[\big]{\abs*{E_{k+1}} - \abs*{E_k}} \ . 
  \]
  Summation gives the bound
  \begin{align*}
    \sum_{k=1}^N \pderiv{\abs{E_k(t)}}{t} &\leq \sum_{k=1}^N \bra*{ \abs[\big]{A_{k-1}[e]} \, c_{k-1} + \abs[\big]{B_k[e]} \, c_k} + \sum_{k=1}^N \abs{ E_k}  \bra[\big]{ A_{k}[d] - A_{k-1}[d]} - \abs{E_N} \, A_{N}[d] \\
    &\qquad + \sum_{k=1}^N \abs{E_k} \bra*{ B_{k-1}[d] - B_{k}[d]} + \abs*{E_{N+1}} \,  B_{N}[d]  \ ,
  \end{align*}
  where $E_0 = 0 = B_0[d]$ by definition. The Assumption~\eqref{e:ass:K2} implies that the kernel grows at most linearly $K(k,l-1)\leq C_K\, k\,l$, from which the estimate
  \[ 
    \abs*{E_{N+1}} \,  B_{N}[d] \leq  C_K \, (1+\varrho) \, N \sum_{j\geq N+1} \bra*{ \abs{c_j} + \abs{d_j}} \leq C_K (1+\varrho) \sum_{j \geq N+1} j \, \bra*{ \abs{c_j} + \abs{d_j}} \to 0 
  \]
  as $N\to \infty$ is obtained. The convergence statement is a consequence of the two conservation laws from Corollary~\ref{cor:convervation_laws} on compact time intervals. 
  
  The terms $A_{k-1}[e]$ and $B_k[e]$ are estimated using the identity
  \begin{align*}
    A_k[e] = \sum_{l\geq 1} K(l,k) \, \bra*{E_{l-1} - E_l} = \sum_{l \geq 1} E_l \, \bra[\big]{ K(l+1,k) - K(l,k)} ,
  \end{align*}
  which by Assumption~\eqref{e:ass:K2} implies
  \[
    \sum_{k=1}^N \abs[\big]{ A_{k-1}[e] } \, c_{k-1} \leq C_K \sum_{l\geq 1} \abs{E_l}  \sum_{k=1}^N k \, c_{k-1} \leq C_K \bra*{1+\varrho} \sum_{l \geq 1} \abs{E_l} \ . 
  \]
  A similar bound applies to $B_k[e]$. Again Assumption~\eqref{e:ass:K2} results in the bound
  \[
     A_{k}[d] - A_{k-1}[d] = \sum_{l\geq 1} \bra[\big]{ K(l,k)-K(l,k-1)} d_l \leq C_K \sum_{l\geq 1} l \, d_l = C_K \varrho
  \]
  and similarly the difference in $B_{k}[d]$. 
  In total, there is some constant $\tilde C = \tilde C(C_K, \varrho)$ such that after passing to the limit $N \to \infty$, the bound
  \[
    \pderiv{}{t} \sum_{k\geq 1} \abs*{E_k(t)} \leq \tilde C \sum_{k\geq 1} \abs*{E_k(t)} 
  \]
  is established, which shows $E_k(t) = 0$ for all $t\in [0,T)$ and $k\geq 1$. This implies that $c_j(t) = d_j(t)$ for all $j\geq 1$ and $t\in [0,T)$. The conservation laws from Corollary~\ref{cor:convervation_laws} imply $\norm{c(t)}= \norm{d(t)}$ and hence $c(t) =d(t)$ on $[0,T)$ in $\cX$ by Proposition~\ref{prop:X:properties}.
\end{proof}
In particular under the refined linear growth Assumption~\eqref{e:ass:K2} the constructed solutions in Theorem~\ref{thm:stability} are unique and generate a semigroup in the sense of Definition~\ref{def:flows}.
\begin{cor}[Semigroup]\label{cor:semigroup}
  Let $K$ satisfy~\eqref{e:ass:K2}, then the solutions to the exchange-driven growth dynamic~\eqref{e:master} are a semigroup on $\cP^\varrho$ in the strong topology for any $\varrho >0$. If in addition $K$ satisfies~\eqref{e:ass:sublinear}, then it constitutes a semigroup on $(\cB^\varrho,\upd)$.
\end{cor}
The uniqueness theorem states that the solution to the truncated system~\eqref{e:master:trunc} converges strongly to the solution of~\eqref{e:master}, whose proof follows along the lines of~\cite[Theorem 3.9]{BCP86} and is omitted.
\begin{prop}\label{prop:strong:approx}
  Let $\bar c \in \cP^\varrho$ for some $\varrho>0$ and suppose that $K$ satisfies~\eqref{e:ass:K2}. Let $c^N$ be the solution of~\eqref{e:master:trunc} with initial data $\bar c^N_k = \bar c_k$ for $k=0,1,\dots, N$. Then as $N\to \infty$ it holds $c^N(t) \to c(t)$ in $\cX$ uniformly on compact time intervals of $[0,\infty)$ with $c$ the unique solution of~\eqref{e:master} on $[0,\infty)$ and $c(0) = \bar c$.
\end{prop}
The last property of solutions to \eqref{e:master} under Assumption~\eqref{e:ass:K2} is their strict positivity for positive times, provided the kernel is strictly positive, too. This result is analog to the one in~\cite[Theorem 4.6]{BCP86}.
\begin{prop}\label{prop:positive}
  Suppose $K: \N\times \N_0 \to [0, \infty)$ is strictly positive and satisfies~\eqref{e:ass:K1}. Let $c$ be a solution to~\eqref{e:master} on some interval $[0,T)$ with $0< T \leq \infty$ such that $c_m(0) >0$ for some $m \geq 1$. Then, it holds $c_k(t) > 0$ for all $k\geq 0$ and $t\in (0,T)$. Moreover, for any $0< t_0 < t < T$, it holds the quantitative lower bound
  \begin{equation}\label{e:lower:bound}
    c_k(t) \geq c_k(t_0) \exp\bra[\big]{ - C_K \, (2\varrho + 1) \, (k+1) \, (t-t_0) } \qquad\text{ for } k\geq 0 \ . 
  \end{equation}
\end{prop}
\begin{rem}\label{rem:vacuum}
  The assumption $c_m(0)>0$ for some $m\geq 1$ is crucial due to the vacuum state $c_k^{\operatorname{vac}} = \delta_{k,0}$, which is an element of $\cP^0$, i.e.~having mass density $\varrho=0$. In particular, the assumption is satisfied for any $c\in \cP^\varrho$ with $\varrho>0$. 
\end{rem}
\begin{proof}
  The proof follows by contradiction. First, the case $c_0(\tau)=0$ for some $\tau \in (0,T)$ is considered. Integrating the equation $\dot c_0(t) = - A_0[c(t)] \, c_0(t) + B_1[c(t)] \, c_1(t)$ on $(0,\tau)$ gives the identity
  \begin{align*}
    0 &= c_0(\tau)\,  \exp\bra*{\int_0^\tau A_0[c(s)] \dx{s}} \\
    &= c_0(0) + \int_0^\tau \exp\bra*{ \int_0^t A_0[c(s)]\dx{s}}\, B_1[c(t)] \, c_1(t) \dx{t} \ . 
  \end{align*}
  Hence, $c_0(0)=0$ and  $B_1[c(t)] c_1(t)=0$ for a.e.~$t\in (0,\tau)$ and by the continuity property from Proposition~\ref{prop:cont:solution} also for all $t\in [0,\tau]$. If $B_1[c(\tau)]=0$, then $c_l(\tau)  = 0$ for all $l\geq 0$ from the positivity of the rates. In the case $B_1[c(\tau)] >0$ follows $c_1(\tau)=0$ and the argument is contained in the case considered below.
  
  Let $c_l(\tau) = 0$ for some $l\geq 1$ and some $\tau \in (0,T)$. The equation
  \begin{align}\label{e:dotc:theta}
    \dot c_l(t)&= \bra*{ A_{l-1}[c(t)] \, c_{l-1}(t) + B_{l+1}[c(t)]\, c_{l+1}(t)} - \theta_l(t) \,c_l(t) \\
    \text{with}\quad \theta_l(t) &= A_l[c(t)] + B_l[c(t)] \ ,
  \end{align}
  integrates to
  \begin{equation}\label{e:zero:p0}
  \begin{split}
    0 &= c_l(\tau) \exp\bra*{ \int_0^\tau \theta_l(s) \dx{s} } \\
    &= c_l(0) + \int_0^\tau \exp\bra*{ \int_0^t \theta_l(s) \dx{s}}\, \bra[\big]{ A_{l-1}[c(t)] \, c_{l-1}(t) + B_{l+1}[c(t)] \, c_{l+1}(t)} \dx{t} \ . 
    \end{split}
  \end{equation}
  Again, it follows that $c_l(0)=0$ and $A_{l-1}[c(t)] c_{l-1}(t) = 0 = B_{l+1}[c(t)] c_{l+1}(t)$ for all $t\in [0,\tau]$. Hence, $A_{l-1}[c(\tau)]=0$ entails $c_k(\tau)=0$ for all $k\geq 1$ or $c_{l-1}(\tau)=0$. Likewise, $B_{l+1}[c(\tau)]=0$ implies $c_k(\tau)= 0$ for all $k\geq 0$ or $c_{l+1}(\tau)=0$. Both cases imply $c_{l-1}(\tau)=0 =c_{l+1}(\tau)$. Then, by induction follows that $c_k(\tau) = 0$ for all $k\geq 0$, which by~\eqref{e:zero:p0} implies $c_k(0)=0$ for all $k\geq 0$. This is a contradiction to the assumption that $c_m(0)>0$ for some $m\geq 1$.
  The lower bound~\eqref{e:lower:bound} follows now from bounding~\eqref{e:dotc:theta} from below. The Assumption~\eqref{e:ass:K1} implies
  \begin{equation}\label{e:unique:p0}
    A_k[d] \leq C_K \, \varrho \, (k+1) \qquad\text{and} \qquad B_k[d] \leq C_K \, \bra*{\varrho +1} \, k \ ,
  \end{equation}
  which in~\eqref{e:dotc:theta} leads to the lower bound. 
  \[
    \dot c_l(t) \geq - \theta_l(t) \, c_l(t) \geq - C_k \, (2\varrho +1 ) \, (k+1) \, c_l(t) \ .
  \]
  The claim~\eqref{e:lower:bound} is now a consequence of the Gronwall Lemma.
\end{proof}
\section{Convergence to equilibrium}\label{s:equilibrium}
\subsection{The Lyapunov function and equilibria}\label{s:lyap}
The goal of this section is to show that all equilibria for the evolution~\eqref{e:master} on $\cP^\varrho$ with $0< \varrho \leq \varrho_c$ are given by $\omega^\varrho$ in~\eqref{def:omega}. The observation that the nonlinear birth death rates~\eqref{e:bdrates} also satisfy a detailed balance condition is useful in this context.
\begin{lem}[Detailed balance for nonlinear birth and death rates]\label{lem:DBC:AB}
  Suppose $K$ satisfies~\eqref{e:ass:K1} and~\eqref{e:ass:BDA}, then for all $\varrho<\infty$ with $0<\varrho\leq \varrho_c$ and $\phi =\phi(\varrho)\in (0,\phi_c]$ uniquely defined through~\eqref{e:def:phi:varrho} holds
  \begin{equation}\label{e:DBC:AB:omega}
    \frac{A_{k-1}[\omega^\varrho] }{B_{k}[\omega^\varrho]} = \frac{\phi(\varrho) \, K(1,k-1)}{K(k,0)} = \frac{\phi(\varrho) \, Q_k}{Q_{k-1}} = \frac{\omega_{k}^\varrho}{\omega_{k-1}^\varrho}\qquad \text{ for all $k\geq 1$} \ .
  \end{equation}
\end{lem}
\begin{proof}
  By the definition~\eqref{def:Q} of $(Q_k)_{k\geq 0}$ holds $K(k,0) \omega_k^\varrho = \phi(\varrho) K(1,k-1)\omega_{k-1}^\varrho$ and hence~\eqref{e:ass:BDA} allows to write
  \begin{align*}
    K(k,0) \, A_{k-1}[\omega^\varrho] &= \sum_{l\geq 1} \frac{K(1,k-1) \, K(k,l-1)}{K(1,l-1)}  \, K(l,0) \,  \omega^\varrho_l \\
    &= \phi(\varrho) \, K(1,k-1) \sum_{l\geq 1} K(k,l-1) \, \omega_{l-1}^\varrho = \phi(\varrho) \, K(1,k-1) \, B_k[\omega^\varrho] \ . 
  \end{align*}
  The result follows from noting that \eqref{def:Q} and \eqref{e:def:omega} imply
  \[
    \omega_k^\varrho / \omega_{k-1}^\varrho = \phi(\varrho) K(1,k-1)/K(k,0) \ .  \qedhere
  \]
\end{proof}
The main tool for the proof of convergence to equilibrium is the Lyapunov function from~\eqref{e:def:Lyapunov}, which is split for the following discussion into an entropy part and potential part
\begin{equation}\label{e:split:Lyapunov}
    \cF[c] = \cS[c] - \sum_{k\geq 0} k\, c_k \log Q_k^{\nicefrac{1}{k}} \qquad\text{with}\qquad \cS[c] = \sum_{k\geq 0} c_k \log c_k  \ . 
\end{equation}
For $\varrho<\infty$ with $0< \varrho \leq \varrho_c$, $\cF$ is actually equivalent to the relative entropy between $c\in \cP^\varrho$ and $\omega^\varrho$, since with $\phi = \phi(\varrho)$ chosen according to~\eqref{e:def:phi:varrho} it holds
\begin{equation}\label{e:RelEnt:F}
  \cH[c | \omega^\varrho] = \sum_{k\geq 0} c_k \log \frac{c_k}{Z(\phi)^{-1} \phi^k Q_k} = \cF[c] + \log Z(\phi) - \varrho \, \log \phi \ .
\end{equation}
The weak$^*$ continuity does not hold in general for $\cF$, but for the relative entropy with respect to the maximal density equilibrium $\omega(\phi_c)$ by an application of the criteria from Proposition~\ref{prop:X:properties}. The following result is analog to \cite[Proposition 4.5]{BCP86} with the only difference that also the number density is fixed. 
\begin{prop}\label{prop:weak:cont:free:energy}
  The relative entropy $\cH[ \cdot | \omega(\phi)]$ is weak$^*$ continuous on $\cB^\varrho$ for any $\varrho>0$ if and only if 
  \begin{equation}\label{e:weak:cont:entropy}
    \lim_{l \to \infty} Q_l^{\nicefrac{1}{l}} = \phi_c^{-1} \in (0,\infty), \qquad Z(\phi_c)<\infty \qquad\text{and}\qquad \phi = \phi_c \ . 
  \end{equation}
\end{prop}
\begin{proof}
  The relative entropy expands to
  \[
    \cH[c | \omega(\phi) ] = \sum_{k\geq 0} c_k \log c_k - \sum_{k=0} c_k \log\bra*{ \phi^k Q_k} + \log Z(\phi) \ .
  \]
  The first entropy term is weak$^*$ continuous on $\cX^+$ by~\cite[Lemma 4.2]{BCP86} and the second is of the form~\eqref{e:potential:energy} with $\alpha_k(\phi) = - \log\bra*{ \phi^k Q_k}$. Hence the sufficient and necessary condition reads \[
    \frac{1}{k}\log\bra*{ \phi^k Q_k} = \log\bra*{ \phi Q_k^{\nicefrac{1}{k}}} \to 0 \qquad\text{if and only if} \qquad \phi = \phi_c \ , 
  \]
  by~\eqref{e:weak:cont:entropy}. Finally, the condition $Z(\phi_c)<\infty$ ensures that the relative entropy is indeed well-defined in this case.
\end{proof}
Lower semicontinuity of the free energy on $\cX^+$ is needed to prove the free energy dissipation relation. In addition, the following result proves continuity in the case $\phi_c\in (0,\infty)$.
\begin{lem}[Strong continuity of the free energy]\label{lem:cF:c}
  Suppose Assumptions~\eqref{e:ass:BDA} and~\eqref{e:ass:Kc} hold.
  \begin{enumerate}
   \item If $\phi_c \in (0,\infty)$, then for any $\varrho>0$ and $c^j \in \cP^\varrho$ such that $c^j \to c$ in $\cX$, it holds
    \begin{equation}\label{e:cF:cont}
      \lim_{j\to \infty} \cF[c^j] = \cF[c] \ . 
    \end{equation}
   \item If $\phi_c = \infty$, then for any $\varrho>0$ and $c^j \in \cP^\varrho$ such that $c^j \to c$ in $\cX$ with $\cF[c]< \infty$, it holds
     \begin{equation}\label{e:cF:cont_lower}
       \liminf_{j\to \infty} \cF[c^j] \geq \cF[c] \ . 
      \end{equation}
      Additionally, for $c\in \cP^\varrho \cap \set{\cF <\infty}$ let $c^N = (c_1,\dots,c_N,0,\dots)$ its truncation at $N\geq 1$. Then it holds
      \begin{equation}\label{e:cF:cont:phi_c}
        \lim_{N\to\infty} \cF[c^N] = \cF[c] \ . 
      \end{equation}
  \end{enumerate}
\end{lem}
\begin{proof}
  The proof uses several times the splitting~\eqref{e:split:Lyapunov} and the result from~\cite[Lemma 4.2]{BCP86} that $\cS$ is finite and weak$^*$ continuous on $\cB^\varrho \subset \cX^+$ for any $\varrho>0$. In particular, it holds $\lim_{j\to \infty} \cS[c^j] = \cS[c]$ if $c^j \to c$ in $\cX$. 
  
  The proof of the (lower semi-)continuity of the second term in the splitting~\eqref{e:split:Lyapunov} splits into the two cases $\phi_c\in (0,\infty)$ and $\phi_c=\infty$. 
  
  \medskip
  
  \noindent\emph{Proof of (1):} The Assumption $\phi_c\in (0,\infty)$ implies that $\lim_{k\to \infty} \log Q_k^{\nicefrac{1}{k}} = - \log \phi_c \in (-\infty,\infty)$ and in particular $\sup_{l\geq 1} \abs[\big]{\log Q_l^{\nicefrac{1}{l}}}<\infty$. Hence, it holds the estimate
  \[
    \lim_{j\to \infty} \abs[\bigg]{ \sum_{k\geq 1} k \, c_k^j \, \log Q_k^{\nicefrac{1}{k}}- \sum_{k\geq 1} k \, c_k \, \log Q_k^{\nicefrac{1}{k}}} \leq \sup_{l\geq 1} \abs[\big]{\log Q_l^{\nicefrac{1}{l}}} \ \sum_{k\geq 1} k \, \abs[\big]{c_k^j - c_k} \to 0 \quad\text{as } j\to \infty \ . 
  \] 
  \noindent\emph{Proof of (2):} The assumption $\cF[c]<\infty$ and the fact that $\cS$ is bounded from below on $\cP^\varrho$ for any $\varrho>0$ by~\cite[Theorem 4.4]{BCP86} imply for any $c\in \cP^\varrho$ the estimate
   \begin{equation}\label{e:cF:cont:p1}
     - \sum_{k\geq 0} k \, c_k \log Q_k^{\nicefrac{1}{k}} \leq \cF[c] - \inf_{c\in \cP^\varrho} \cS[c] < \infty  \ .
   \end{equation}
   The Assumption~\eqref{e:ass:Kc} yields the estimate $\sup_{k\geq 0} \max\set[\big]{ \log Q_k^{\nicefrac{1}{k}},0 } = \log \overbar Q <\infty$ for some $\overbar Q\geq 1$, implying the bound
   \begin{equation}\label{e:cF:cont:p2}
     - \sum_{k\geq 0} k \, c_k \log Q_k^{\nicefrac{1}{k}} \geq - \varrho \log \overbar Q  > -\infty \ . 
   \end{equation}
   Hence, for any $\eps>0$, there exists $M_1=M_1(\eps)$ such that $\abs[\big]{\sum_{k\geq M_1} k \, c_k \log Q_k^{\nicefrac{1}{k}}} \leq \frac{\eps}{4}$. Since $c^j \to c$ in $\cX$, there exists $M_2=M_2(\eps)$, such that $\sum_{k\geq M_2} k \, c_k^j \leq \frac{\eps}{4 \log \overbar Q}$ uniformly in $j$. 
   Altogether, for any $M\geq \max\set{M_1,M_2}$ holds the bound
   \begin{align*}
     - \sum_{k\geq 0} k\, \bra[\big]{ c_k^j - c_k} \log Q_k^{\nicefrac{1}{k}} &\geq - \sum_{k=0}^{M-1} k\, \bra[\big]{ c_k^j - c_k} \log Q_k^{\nicefrac{1}{k}} -  \log \overbar Q \sum_{k\geq M} k \, c_k^j  - \abs[\bigg]{\sum_{k\geq M} k \, c_k \log Q_k^{\nicefrac{1}{k}}} \\
     &\geq - \sup_{0\leq l \leq M-1} \abs[\big]{\log Q_l^{\nicefrac{1}{l}}} \ \sum_{k=0}^{M-1} k \, \abs[\big]{c_k^j - c_k} - \frac{\eps}{2} \ . 
   \end{align*}
   Now, using once more $c^j \to c$ in $\cX$, there exists $j$ large enough such that the first term becomes bounded from below by $-\eps /2$.
   
   Since $c^N \to c$ in $\cX$, it also holds $\cS[c^N]\to \cS[c]$ as before. By the splitting~\eqref{e:split:Lyapunov}, the result follows from~\eqref{e:cF:cont:p1} and~\eqref{e:cF:cont:p2} by noting that
   \[
     \abs*{ \cF[c^N] - \cF[c]} \leq \abs*{\cS[c^N] - \cS[c]} + \abs[\bigg]{ \sum_{k\geq N+1} k \, c_k \log Q_k^{\nicefrac{1}{k}}} \to 0 \qquad\text{as } N\to \infty \ .
   \]
\end{proof}
\begin{prop}[Free energy dissipation relation]\label{prop:EED}
  Suppose the Assumptions~\eqref{e:ass:K1}, \eqref{e:ass:BDA} and~\eqref{e:ass:Kc} hold.
  Let $c$ be the solution constructed in Theorem~\ref{thm:stability} to some initial data $\bar c \in \cP^\varrho$ for some $\varrho>0$ with $\cF[\bar c] < \infty$. Then, for all $t\in [0,T)$ it holds 
  \begin{equation}\label{e:FED}
    \cF[c(t)] + \int_0^t \cD[c(s)] \dx{s} \leq \cF[\bar c] \ ,
  \end{equation} 
  where the dissipation is given by
\begin{equation}\label{e:def:dissipation}
  \cD[c] = \frac{1}{2} \sum_{k\geq 1} \sum_{l\geq 1} \psi_{\Boltz}\bra[\big]{ K(k,l-1) \, c_k \, c_{l-1}, K(l,k-1) \, c_{l} \, c_{k-1}} 
\end{equation}
with $\psi_{\Boltz}(a,b) = \bra*{a-b}\bra*{ \log a - \log b}$. 
\end{prop}
\begin{proof}
  For $N\geq 1$, the truncated Lyapunov function and  dissipation are defined by
  \begin{align*}
    \cF^N[c] &= \sum_{k=0}^N c_k \log \frac{c_k}{Q_k} \ , \\
    \cD^N[c] &= \frac{1}{2} \sum_{k=1}^N \sum_{l=1}^N \psi_{\Boltz}\bra[\big]{K(k,l-1) \, c_k \, c_{l-1}, K(l,k-1) \, c_l \, c_{k-1}} \ . 
  \end{align*}
  Let $c^N$ be the solution of the truncated system~\eqref{e:master:trunc} to the truncated initial data $\bar c^N_k = \bar c_k$ for $k=0,1,\dots, N$.
  Since $\bar c\in \cP^\varrho$ for some $\varrho>0$, there exists $\bar c_m >0$ for some $m\geq 1$ and hence also $\bar c_m^N>0 $ for $N$ sufficiently large. By the same argument as in Proposition~\ref{prop:positive}, it holds that $c_k^N(t_0)>0$ for all $k = 0,1,\dots , N$ and any $t_0 >0$. From here the free energy dissipation relation for the truncated system can be calculated by using that the relation~\eqref{e:def:net:flux} holds similarly for the truncated system, a summation by parts noting that $J_{-1}^N[c]= 0 = J_N^N[c]$ and symmetrization of the sum
  \begin{align*}
    \pderiv{}{t} \cF^N[c^N] &= \sum_{k=0}^N \log \frac{c_k^N}{Q_k}  \dot c_k = \sum_{k=0}^N \log \frac{c_k^N}{Q_k} \bra*{J_{k-1}^N[c] - J_k^N[c]} \\
    &= - \sum_{k=1}^N \sum_{l=1}^N  \bra*{ \log \frac{c_{k-1}^N}{Q_{k-1}} - \log \frac{c_{k}^N}{Q_{k}}} \bra*{ j_{l,k-1}[c^N] - j_{k,l-1}[c^N]} \\
    &= - \frac{1}{2} \sum_{k=1}^N \sum_{l=1}^N \bra*{ \log \frac{c_{k-1}^N c_l^N}{Q_{k-1}Q_l} - \log \frac{c_{k}^N c_{l-1}^N}{Q_{k}Q_{l-1}}}  \bra*{ j_{l,k-1}[c^N] - j_{k,l-1}[c^N]} \\
    &= - \frac{1}{2} \sum_{k=1}^N \sum_{l=1}^N \psi_{\Boltz}\bra*{ j_{l,k-1}[c^N], j_{k,l-1}[c^N]} = - \cD^N[c^N] \ .
  \end{align*}
  Hereby, the last identity is a consequence of~\eqref{e:Q:DBC}. Now, along the subsequence $(N_n)$, for which $c^{N_n} \weakto c$ as in the proof of Theorem~\ref{thm:stability}, holds the energy estimate
  \[
    \cF[c^{N_n}(t)] + \int_0^t \cD^{N_n}[c^{N_n}(s)] \dx{s} = \cF[\bar c^{N_n}] \ .
  \]
  Since $\cF[\bar c]< \infty$ and $\bar c^{N_n}\to \bar c$ in $\cX$, Lemma~\ref{lem:cF:c} implies $\lim_{n\to \infty} \cF[\bar c^{N_n}] = \cF[\bar c]$. Since $\cD^{N}[c^{N}] \geq \cD^{m}[c^N]$ for all $m< N$ is holds 
  \[
    \liminf_{n\to \infty} \int_0^t \cD^{N_n}[c^{N_n}(s)] \dx{s} \geq \int_0^t \cD[c(s)] \dx{s} \ . 
  \]
  Finally, by Corollary~\ref{cor:convervation_laws} and Proposition~\ref{prop:trunc:system} follows for all $t>0$ and all $N$
  \[
    \norm{c^N(t) } = \sum_{k\geq 0} (k+1) \, c_k^N(t) = \varrho + 1  = \sum_{k\geq 0} (k+1) \, c_k(t) = \norm{c(t)} \ .
  \]
  Hence, $c^N(t) \xrightharpoonup{*} c(t)$ and $\norm{c^N(t)} = \norm{c(t)}$ for all $t\geq 0$, which by Proposition~\ref{prop:X:properties} implies that $c^N(t) \to c(t)$ in $\cX$. Finally, Lemma~\ref{lem:cF:c} yields that $\liminf_{N\to \infty} \cF[c^N(t)] \geq \cF[c(t)]$ as $N\to \infty$ for all $t\in (0,\infty)$.
\end{proof}
\begin{prop}[Stationary states]\label{prop:stationary}
  Suppose the Assumptions~\eqref{e:ass:K2}, \eqref{e:ass:BDA} and~\eqref{e:ass:Kc} hold. Then, the stationary states are characterized by:
  \begin{enumerate}
   \item For $\varrho <\infty$ with $0< \varrho \leq \varrho_c$, the equilibrium state $\omega^\varrho$ defined in~\eqref{def:omega} are unique on $\cP^\varrho\cap \set{\cF < \infty}$. 
   \item For $\varrho_c< \varrho < \infty$, there exists no equilibrium state on $\cP^\varrho\cap \set{\cF < \infty}$.
  \end{enumerate}
\end{prop}
\begin{proof}
  The Assumption~\eqref{e:ass:K2} entails that the system~\eqref{e:master} has a unique solution by Theorem~\ref{thm:unique}, which additionally satisfies the free energy dissipation relation~\eqref{e:FED} by Proposition~\ref{prop:EED}. Hence $c^*$ is a stationary state if and only if $\cD[c^*] = 0$. The definition of the dissipation~\eqref{e:def:dissipation} shows that $c^*$ satisfies~\eqref{e:Q:DBC}, since $\psi_{\Boltz}(a,b)=0$ if and only if $a=b$. The statement of the proposition follows from the construction of $\omega^\varrho$ in~\eqref{e:def:omega}. 
\end{proof}
\begin{thm}[{Free energy minimizer \cite[Theorem 4.4]{BCP86}}]
  Suppose the Assumptions \eqref{e:ass:K1}, \eqref{e:ass:BDA} and~\eqref{e:ass:Kc} hold.
\begin{itemize}
 \item  Let $\varrho< \infty$ and $0\leq \varrho \leq \varrho_s$. Then $\omega^\varrho$ defined in~\eqref{def:omega} is the unique minimizer of $\cF[c]$ and $\cH[c| \omega^\varrho]$ over $\cP^\varrho$ and every minimizing sequence converges strongly to $\omega^\varrho$ in $\cX$.
 \item Let $\varrho > \varrho_c$. Then 
 \[
   \inf_{c\in \cP^\varrho} \cH[ c | \omega^{\varrho_c}] = 0 
 \]
 and any minimizing sequence converges weak$^*$ to $\omega^{\varrho_c}$ in $\cX$, but not strongly.
\end{itemize}
\end{thm}

\subsection{Relative compactness of trajectories}\label{s:compact}

The starting point is that the solution to~\eqref{e:master} is a generalized flow in strong topology (Theorem~\ref{thm:generalized_flow} under Assumption~\eqref{e:ass:K1}) or a generalized flow in the weak topology (Theorem~\ref{thm:generalized_flow:weak} under Assumption~\eqref{e:ass:sublinear}) and constitutes in both cases a semigroup (Corollary~\ref{cor:semigroup} under Assumption~\eqref{e:ass:K2}). The dissipative nature of the evolution is captured by the free energy dissipation relation~\eqref{e:FED} providing a Lyapunov function for the evolution. If relative compactness of the orbits in $\cX$ for the according topology is proven, then the longtime limit can be deduced by the following invariance principle.
\begin{prop}[Invariance principle {\cite[Proposition 5.3]{BCP86}}]\label{prop:invariance_principle}
  Let $G$ be a generalized flow (Definition~\ref{def:flows}) on some metric space $(\cY,d)$. Let $\varphi(\cdot)\in G$ and suppose that its positive orbit $\mathcal{O}^+(\varphi) = \bigcup_{t\geq 0} \varphi(t)$ is relatively compact. Then 
  \[
    \varOmega(\varphi) = \set*{ \varPhi \in \cY: \varphi(t_j) \to \varPhi \text{ for some sequence $t_j\to \infty$ as $j\to \infty$} }
  \]
  is nonempty and satisfies
  \[
    d\bra*{\varphi(t),\varOmega(\varphi)} \to 0 \quad\text{as}\quad t \to \infty \ .
  \]
  Moreover $\varOmega(\varphi)$ is quasi-invariant, that is for any $\varPhi \in \varOmega(\varphi)$ there exists $\varphi(\cdot)\in G$ with $\varphi(0) = \varPhi$ and $\mathcal{O}^+(\varphi) \subset \varOmega(\varphi)$.
\end{prop}
The relative compactness in the strong topology can be easily deduced in the case, where the radius of convergence in~\eqref{e:ass:Kc} is infinite. 
\begin{lem}[Relative compactness for $\phi_c = \infty$]\label{lem:relcompact:infty}
  Suppose Assumptions~\eqref{e:ass:BDA} and \eqref{e:ass:Kc} with $\phi_c = \infty$ hold. Let $\varrho>0$ and $c \in \cP^{\varrho}$ be any solution of~\eqref{e:master} on $[0, \infty)$ satisfying $c(0) \ne 0$, $\cF[c(0)] < \infty$ and the free energy dissipation relation~\eqref{e:FED}. Then $\bra{c(t)}_{t\geq 0}$ is relatively compact in $\cP^{\varrho}$, that is for any $\eps>0$ exists $M=M(\eps)$ such that
  \begin{equation}\label{e:tight:1}
    \sup_{t\geq 0} \sum_{k\geq M(\eps)} k \, c_k(t) \leq \eps \ . 
  \end{equation}
\end{lem}
\begin{proof}
  The free energy dissipation relation~\eqref{e:FED} implies $\cF[c(t)]\leq \cF[c(0)]< \infty$ for all $t\geq 0$. The entropy $\cS$ is bounded from below on $\cB^{\varrho}$ by~\cite[Lemma 4.2]{BCP86} entailing
  \[
    - \sum_{k\geq 1} k c_k(t) \log Q_k^{\nicefrac{1}{k}} \leq \cF[c(0)]- \inf_{c\in \cB^\varrho} \cS[c] = C < \infty \qquad\text{for all } t\geq 0 \ .
  \]
  Since $\phi_c=\infty$, it holds $\lim_{k\to \infty} Q_k^{\nicefrac{1}{k}} = 0$. Hence, for any $\eps>0$, there exists $M=M(\eps)$ such that 
  \[
    -\log Q_k^{\nicefrac{1}{k}} \geq \frac{C + \varrho \log \overbar Q}{\eps} \quad \text{ for all } k\geq M \qquad\text{with}\qquad \overbar Q = \sup_{k\geq 0} Q_k^{\nicefrac{1}{k}} \geq 1 \ . 
  \]
  This estimate implies the bound
  \[\begin{split}
    \sum_{k\geq M} k \, c_k(t) &\leq \frac{\eps}{C+ \varrho \log \overbar Q} \bra*{ - \sum_{k\geq 1} k \, c_k(t) \log Q_k^{\nicefrac{1}{k}} + \sum_{k=1}^{M-1} k \, c_k(t) \max\set*{ \log Q_k^{\nicefrac{1}{k}},0}} \ .
  \end{split}\] 
  The first sum is bounded by $C$ and the second by $\log \overbar Q$ concluding~\eqref{e:tight:1}.
\end{proof}
Lemma~\ref{lem:relcompact:infty} together with Corollary~\ref{cor:semigroup} and Proposition~\ref{prop:invariance_principle} establish the proof of Theorem~\ref{thm:longtime} in the case $\phi_c=\infty$. 
\begin{cor}[Longtime behavior for $\phi_c=\infty$]\label{cor:longtime:phc:infty}
   If Assumptions~\eqref{e:ass:BDA}, \eqref{e:ass:K2} and \eqref{e:ass:Kc} with $\phi_c = \infty$ hold, then for any $\varrho_0\in [0,\infty) $ and any $\bar c\in \cP^{\varrho_0}$ with $\cF[\bar c] < \infty$ the unique solution $c$ of~\eqref{e:master} with $c(0) = \bar c$ satisfies $c(t) \to \omega^{\varrho}$ strongly in $\cX$ as $t \to \infty$.
\end{cor}
\begin{proof}
  The case $\varrho_c=\infty$ follows immediately by the relative compactness statement of Lemma~\ref{lem:relcompact:infty} in combination with the invariance principle from Proposition~\ref{prop:invariance_principle}. The strong convergence implies also the continuity statement for the free energy in the limit by Lemma~\ref{lem:cF:c}. 
\end{proof}
Before turning to the more involved proof of relative compactness in the strong topology of solutions to~\eqref{e:master} in $\cP^\varrho$ for $0<\phi_c < \infty$, a weak$^*$ convergence result is stated. Since the result relies on the semigroup in weak$^*$ topology, the sublinear growth assumption~\eqref{e:ass:sublinear} is needed. 
This result is an immediate consequence of the generalized flow in the weak$^*$ topology from Theorem~\ref{thm:generalized_flow:weak} and the free energy dissipation relation~\eqref{e:FED}. It is the analog to~\cite[Theorem 5.5]{BCP86} and~\cite[Theorem 5.10]{Slemrod1989} for the Becker-Döring system.
\begin{thm}\label{thm:longtime:weak}
  Suppose Assumptions~\eqref{e:ass:sublinear} and~\eqref{e:ass:Kc} hold with $0 < \phi_c <\infty $. For $\varrho_0>0$ and $\bar c\in \cP^{\varrho_0}$ and let $(c_t)_{t\geq 0}\subset \cP^{\varrho_0}$ be a solution of~\eqref{e:master} to $\bar c$ satisfying the free energy dissipation relation~\eqref{e:FED}. Then $c(t) \xrightharpoonup{*} \omega^\varrho$ as $t\to \infty$ for some $0 \leq \varrho\leq \min\set*{ \varrho_0 , \varrho_c}$. 
\end{thm}
\begin{proof}
  For $\varrho_0 = 0$, the only possible state is the vacuum state from Remark~\ref{rem:vacuum}. Let $\varrho_0>0$, then Theorem~\ref{thm:generalized_flow:weak} yields that the solutions generate a generalized flow on~$(\cB^\varrho,d)$. The relative entropy $\cH[\cdot | \omega(\phi_c)]$ is weak$^*$ continuous on $\cB^\varrho$ by Proposition~\ref{prop:weak:cont:free:energy}. The conservation laws from Corollary~\ref{cor:convervation_laws} give uniform bounds $\norm{c(t)} \leq 1+\varrho_0$ for all $t\geq 0$ implying that $\cO^+(c)$ is relatively compact in $(\cB^{\varrho_0},\upd)$. In addition any weak$^*$ limit point will always satisfy the conservation law $\sum_{k\geq 0} c_k = 1$, since the bounded first moment provides the necessary relative compactness in $\ell^1(\N_0)$.
  
  By Proposition~\ref{prop:invariance_principle} follows that $\Omega(c)$ is nonempty and consists of solutions $c(\cdot)$ along which $\cH[c(t)|\omega(\phi_c)]$ has the constant value $h^\infty$. Applying the free energy dissipation~\eqref{e:FED} to such solutions gives the identity
  \[
    h^\infty + \int_0^t \cD[c(s)] \dx{s} \leq h^\infty \ . 
  \]
  Hence, $\cD[c(s)] = 0$ for a.e. $s\in (0,t)$ by the nonnegativity~\eqref{e:def:dissipation} of $\cD$. The form of $\cD$ implies that for any fixed $s\in (0,t)$ it holds $c(s) = \omega^{\varrho}$ for some $0\leq \varrho\leq \min\set*{ \varrho_0,\varrho_c}$. Hence, $\Omega(c)$ consists of the states $\omega^\varrho$ with $0\leq \varrho\leq \min\set*{ \varrho_0 , \varrho_c}$. The unique state is identified by considering
  \[
    \cH[\omega^\varrho | \omega^{\varrho_c}] = \sum_{l\geq 0} \omega_l^\varrho \log\frac{\phi(\varrho)^l \, Z(\phi(\varrho_c))}{\phi(\varrho_c)^l \, Z(\phi(\varrho))}
    = \varrho \log\frac{\phi(\varrho)}{\phi(\varrho_c)} + \log \frac{Z(\phi(\varrho_c))}{Z(\phi(\varrho))} \ .
  \]
  The identity $\pderiv{}{\phi} \log Z(\phi) = \frac{\varrho(\phi)}{\phi}$ gives
  \[
    \pderiv{}{\varrho} \cH[\omega^\varrho | \omega^{\varrho_c}] = \frac{\phi(\varrho)}{\phi(\varrho_c)} + \varrho \frac{\phi'(\varrho)}{\phi(\varrho)} - \varrho \frac{\phi'(\varrho)}{\phi(\varrho)} = \frac{\phi(\varrho)}{\phi(\varrho_c)}  \ .
  \]
   The mapping $\varrho \mapsto \cH[\omega^\varrho | \omega^{\varrho_c}]$ is one-to-one on $[0,\varrho_c)$, since $\varrho \mapsto \phi(\varrho)$ is one-to-one on $[0,\varrho_c)$ by~\eqref{e:phi:strict:monotone}. Hence, the equation $\cH[\omega^\varrho | \omega^{\varrho_c}] = h^\infty$ has a unique solution for some $\varrho$ with $0\leq \varrho \leq \min\set*{\varrho_0,\varrho_c}$. 
\end{proof}
The proof of relative compactness for orbits in the case $0 < \phi_c < \infty$ is based on the same strategy as in~\cite{BC88}, which was also successfully applied to generalized~\cite{Canizo05} and modified Becker-Döring systems~\cite{HNN06}, and macroscopic limits of the Becker-Döring system~\cite{LM02}. The crucial idea is to consider the new variable
\begin{equation}\label{e:def:xn}
  x_l(t) = \sum_{k\geq l} k \, c_k(t) 
\end{equation}
for which the uniform in time bound $\frac{x_l(t)}{\lambda_l} \lesssim C$ is established, with $\lambda_l \to 0$ as $l \to \infty$. This estimate yields the relative compactness of trajectories in $\cX$. The following proposition provides a tightness result conditioned on certain estimates satisfied by the nonlinear birth-death rates~\eqref{e:bdrates}. In a second step, it will be ensured that these estimates actually hold thanks to the weak$^*$ convergence from Theorem~\ref{thm:longtime:weak}.
\begin{prop}[Relative compactness ($0< \phi_c < \infty$)]\label{prop:tight}
  Suppose Assumption~\eqref{e:ass:K2} and~\eqref{e:ass:Kc} hold with $0 < \phi_c <\infty $.
  Let $\bar c \in \cP^\varrho$ for some $\varrho>0$ and $c:[0,\infty) \to \cP^\varrho$ be the unique solution of~\eqref{e:master} with $c(0) = \bar c$. Suppose that for some $\phi < \phi_c$ there exists $l_0$ such that
  \begin{equation}\label{e:ass:ABphi}
    \sup_{t\geq 0} \frac{A_{l}[c(t)]}{B_{l}[c(t)]} < 1 \quad\text{and}\quad  \sup_{t\geq 0} \frac{A_{l}[c(t)]}{B_{l+1}[c(t)]} < \frac{\phi K(1,l)}{K(l+1,0)} \quad\text{uniformly in $l\geq l_0$.}
  \end{equation}
  Let $(\lambda_l)$ be a positive nonincreasing sequence satisfying
  \begin{equation}\label{e:def:lambda}
     \lambda_l - \lambda_{l+1}\geq  \nu_l \, \bra*{ \lambda_{l-1}- \lambda_l } \quad\text{with}\quad \nu_l = \frac{l^2}{(l-1)^2} \frac{\phi K(1,l-1)}{K(l,0)} \ . 
  \end{equation}
  Then it holds that
  \begin{equation}\label{e:def:H}
    H(t) = \max\set*{ \sup_{l \geq l_0 + 1} \frac{x_l(t)}{\lambda_l}, \frac{\varrho}{\lambda_{l_0}}} \qquad\text{ is nonincreasing on $[0,\infty)$. } 
  \end{equation}
\end{prop}
\begin{proof}
  By the uniqueness of the trajectory for the initial datum $\bar c$, it suffices to prove that if $H(0)<\infty$ and given $T>0$, then for any $\eps >0$ holds $H(t) \leq H(0) + \eps$ for $t\in [0,T]$. It is again more convenient to prove the result with the help of the truncated system~\eqref{e:master:trunc}. Likewise, let $\bar c^{N}$ denote the truncated initial data to $\bar c$ and let $\varrho^N = \sum_{k= 1}^N k \bar c_k$. The strong convergence $\bar c^{N}\to \bar c$ in $\cX$ implies the strong convergence $c^{N} \to c $ in $C([0,T]; \cX)$ for any $T\in (0,\infty)$ by the uniqueness assumption and Theorem~\ref{thm:generalized_flow}. Likewise, the truncated birth and death rates~\eqref{e:master:trunc:birth-death} satisfy $A_{l}^N[c^N] \to A_l[c]$ and $B_{l+1}^N[c^N] \to B_{l+1}[c]$ on $C([0,T])$ as $N\to \infty$ for all $l\geq 0$ by the Assumption~\eqref{e:ass:K1}. In particular, by Assumption~\eqref{e:ass:ABphi} there is $\N_0 \geq 1$ such that for all $N\geq \N_0$, $l_0 \leq l \leq N$ and $t \in [0,T]$ it holds 
  \begin{equation}\label{e:ABphi:p0}
    \frac{A_{l}^N[c^N(t)]}{B_{l}^N[c^N(t)]} < 1 \qquad\text{and}\qquad \frac{A_{l}^N[c^N(t)]}{B_{l+1}^N[c^N(t)]} < \frac{\phi K(1,l)}{K(l+1,0)}  \ .
  \end{equation}
  Let for $N\geq \N_0$
  \[
    y_l(t) = \sum_{k\geq l} k \, c_k^N(t) \ , \quad g(t) = \sup_{l\geq l_0+1} \frac{y_l(t)}{\lambda_l} \quad\text{and}\quad
    H^N(t) = \max\set*{g(t) , \frac{\varrho^N}{\lambda_{l_0}}} \ .
  \]
  The main step consists in proving that 
  \begin{equation}\label{e:HN}
    H^N(t) \leq H^N(0) + \eps \qquad\text{for all} \qquad t\in [0,T] \ . 
  \end{equation}
  Suppose the contrary holds. Since $H^N$ is absolutely-continuous, there exists $s\in [0,T]$ such that $H^N(s) = K_\eps = H^N(0) + \eps$. Since $H^N(0) \geq \varrho^N / \lambda_{l_0}$, it holds $g(s) = K_\eps$, which implies that $y_l(s) / \lambda_l = K_\eps$ for some minimal $l$ with $l_0+1 \leq l \leq N$ such that
  \begin{equation}\label{e:choice:ym}
    \frac{y_{l-1}(s)}{\lambda_{l-1}} < K_\eps, \qquad \frac{y_{l+1}(s)}{\lambda_{l+1}} \leq K_\eps \qquad\text{and}\qquad \dot y_l(s) \geq 0 \ .
  \end{equation}
  From the definition of $y_l$ and~\eqref{e:master:trunc} it holds
  \begin{align*}
    \dot y_l &= \sum_{k=l}^N J_k^N[c^N] + l \, J_{l-1}^N[c^N] \\
    &= \sum_{k= l+1}^N \bra*{A_k^N[c^N] - B_k^N[c^N]} \, c_k^N + A_l^N[c^N]  \, c_l^N + l\, \bra*{ A_{l-1}^N[c^N]\,c_{l-1}^N - B_{l}^N[c^N] \, c_l^N} \ .
  \end{align*}
  Since $A_k^N[c^N] - B_k^N[c^N] \leq 0$ by \eqref{e:ABphi:p0} for $k\geq l_0$, the above identity is bounded by
  \begin{align*}
    \dot y_l &\leq A_l^N[c^N] \, c_l^N + l\, \bra*{ A_{l-1}^N[c^N]\, c_{l-1}^N - B_{l}^N[c^N]\, c_l^N} \\
    &= l \, A_{l-1}^N[c^N] \, \frac{y_{l-1} - y_l}{l-1} - \underbrace{\bra*{ l \, B_l^N[c^N] - A_l^N[c^N]}}_{\geq (l-1) \, B_l^N[c^N]} \frac{y_l - y_{l+1}}{l} \\
    &\leq \frac{l-1}{l} \, B_l^N[c^N] \, \pra*{ \frac{l^2}{(l-1)^2} \, \frac{A_{l-1}^N[c^N]}{B_l^N[c^N]} \, \bra*{ y_{l-1} - y_l } - \bra*{ y_l - y_{l+1}}} \\
    &< K_\varepsilon \frac{l-1}{l} B_l^N[c^N] \, \pra*{\frac{l^2}{(l-1)^2} \frac{A_{l-1}^N[c^N]}{B_l^N[c^N]} (\lambda_{l-1} - \lambda_l) - (\lambda_l - \lambda_{l+1}) } \\
	&\le K_\varepsilon \frac{l-1}{l} B_l^N[c^N] \, \pra*{ \frac{l^2}{(l-1)^2} \frac{\phi K(1,l-1)}{K(l,0)} (\lambda_{l-1} - \lambda_l) - (\lambda_l - \lambda_{l+1}) } \\
    &\le K_\eps \, \frac{l-1}{l} \,  B_l^N[c^N] \, \pra[\big]{ \nu_l \, \bra*{ \lambda_{l-1} - \lambda_l } - \bra*{ \lambda_l - \lambda_{l+1}}} \leq 0 \ .
  \end{align*}
  In the last estimates, the estimate~\eqref{e:ABphi:p0} was applied, the choice~\eqref{e:choice:ym} of $y_l$ was used and also that $(\lambda_k)_{k\geq 0}$ satisfies~\eqref{e:def:lambda} and is moreover monotone. Hence, $\dot y_l < 0$, which is contradiction and proves~\eqref{e:HN}. The result follows from letting $N\to \infty$ using the strong convergence of the truncation $c^N \to c$ in $C([0,T];\cX)$ of Proposition~\ref{prop:strong:approx}.
\end{proof}
The invariance principle from Theorem~\ref{thm:longtime:weak} implies that $c(t) \xrightharpoonup{*} c^\varrho$. This information combined with the following proposition yields the strong convergence in the case $\varrho < \varrho_c$.
\begin{prop}\label{prop:weak_to_strong}
  Suppose that Assumption~\ref{ass:BDA} holds.  Then any solution $(c(t))_{t\geq 0}$ to~\eqref{e:master} with $c(t) \xrightharpoonup{*} c^\varrho$ as $t \to \infty$ for some $\varrho < \varrho_c$ satisfies $c(t) \to c^\varrho$ strongly in $\cX$. 
\end{prop}
\begin{proof}
  The proof concludes in two steps. In the first step, the weak convergence $c(t) \xrightharpoonup{*} c^\varrho$ for some $\varrho < \varrho_c$ implies that there exists $t_0$ large enough such that~\eqref{e:ass:ABphi} hold for all $t\geq t_0$. This ensures that Proposition~\ref{prop:tight} can be applied in a second step.
       
  Now, suppose $c(t) \xrightharpoonup{*} \omega^\varrho$ as $\to \infty$ with $\varrho<\varrho_c$.
  
  \textbf{Step~1.} Let $\eps>0$. It readily follows from the identity~\eqref{e:DBC:AB:omega}, \eqref{e:ass:Kc}, and~\eqref{e:ass:K3} that there exists $l_1=l_1(\eps)$ such that
  \begin{align}\label{e:AB:omega_comp1a}
    \frac{A_{l-1}[\omega^\varrho]}{B_l[\omega^\varrho]} & = \frac{\phi(\varrho) K(1,l-1)}{K(l,1)} = \frac{\phi(\varrho) K(1,l-1)}{K(l,0)} \frac{K(l,0)}{K(l,1)} \notag \\
    & \le (1+\eps)\frac{\phi(\varrho) K(1,l-1)}{K(l,0)} \quad\text{ for all }\quad l\ge l_1(\eps) 
  \end{align}
  and
  \begin{equation}\label{e:AB:omega_comp1b}
    \frac{K(1,l-1)}{K(l,0)} \le \frac{1+\eps}{\phi_c}\quad\text{ for all }\quad l\ge l_1(\eps) \ .
  \end{equation}
  Another consequence of~\eqref{e:ass:K3} is that there is $l_2=l_2(\eps)$ such that, for all $t\ge 0$, 
  \begin{align}
    \frac{A_l[c(t)]}{A_{l-1}[c(t)]} & = \frac{\sum_{k\ge 1} \frac{K(k,l)}{K(k,l-1)} K(k,l-1) c_k(t)}{\sum_{k\ge 1} K(k,l-1) c_k(t)} \le \sup_{k\ge 1} \frac{K(k,l)}{K(k,l-1)} \notag \\
    & \le 1+\eps \qquad\text{ for all }\qquad l\ge l_2(\eps)\,, \ t\ge 0.  \label{e:AB:omega_comp2}
  \end{align}
  Now, assume that the following bound is already established:
  \begin{equation}
    \max\left\{ \frac{A_{l-1}[c(t)]}{A_{l-1}[\omega^\varrho]} , \frac{B_l[\omega^\varrho]}{B_l[c(t)]} \right\} \le 1+ \eps\,, \qquad l\ge 1\,, \ t\ge t_0(\eps)\,. \label{e:AB:weak_comp}
  \end{equation}
  From~\eqref{e:AB:weak_comp}, the claim~\eqref{e:ass:ABphi} is deduced. Indeed, by a combination of the estimates~\eqref{e:AB:omega_comp1a} and \eqref{e:AB:weak_comp} for any $t\ge t_0(\eps)$ and $l\ge l_0(\eps)=\max\{l_1(\eps), l_2(\eps)\}$ follows
  \begin{align}
    \frac{A_{l-1}[c(t)]}{B_l[c(t)]} & = \frac{A_{l-1}[c(t)]}{A_{l-1}[\omega^\varrho]} \frac{B_l[\omega^\varrho]}{B_l[c(t)]} \frac{A_{l-1}[\omega^\varrho]}{B_l[\omega^\varrho]} \notag \\
    & \le (1+\eps)^3 \frac{\phi(\varrho) K(1,l-1)}{K(l,0)} = \frac{\tilde{\phi} K(1,l-1)}{(1+\eps)^2 K(l,0)} \,, \label{e:AB:weak_comp_b}
  \end{align}
  with $\tilde{\phi}=(1+\eps)^5 \phi(\varrho)$. By choosing$\eps>0$ small enough so that $\tilde{\phi}<\phi_c$, the second condition in~\eqref{e:ass:ABphi} holds, from which the first follows as
  \begin{align*}
    \frac{A_l[c(t)]}{B_l[c(t)]}  = \frac{A_l[c(t)]}{A_{l-1}[c(t)]} \frac{A_{l-1}[c(t)]}{B_l[c(t)]} \le \frac{\tilde{\phi}}{(1+\eps)K(l,0)} \le \frac{\tilde{\phi}}{\phi_c}< 1
  \end{align*}
  by~\eqref{e:AB:omega_comp1b}, \eqref{e:AB:omega_comp2} and \eqref{e:AB:weak_comp_b}.

  Hence, $(c_k(t))_{k\geq 0,t \geq 0}$ satisfies~\eqref{e:ass:ABphi} for some $\tilde \phi < \phi_c$, all $t\geq t_0$ and $l\geq l_0$, which finishes step 1, once~\eqref{e:AB:weak_comp} is proven.
  
  \textbf{Proof of~\eqref{e:AB:weak_comp}:} To do so, by the weak$^*$ convergence and strict positivity of $\omega^\varrho$, there exists for any $\eps>0$ and any $M\geq 1$ a $t_1=t_1(\eps,M)$ such that $c_l(t) \leq \bra*{1 + \frac{\eps}{2}}\omega^\varrho_l$ for all $l=0,\dots , M-1$ and all $t\geq t_1$. Then, the Assumption~\eqref{e:ass:K4} implies
  \begin{align*}
    \frac{A_{l-1}[c(t)]}{A_{l-1}[\omega^\varrho]} &= \frac{\sum_{k=1}^{M-1} K(k,l-1)\, c_k(t) + \sum_{k\geq M} K(k,l-1) \, c_k(t)}{\sum_{k\geq 1} K(k,l-1) \, \omega^\varrho_k} \\
    &\leq 1+ \frac{\eps}{2} + \frac{C_K \, a_{l-1} \sum_{k\geq M}  d_k \, c_k(t)}{C_K^{-1} \, a_{l-1} \sum_{k\geq 0} \omega^\varrho_k} \leq 1 + \frac{\eps}{2} + C_K^2 \frac{d_M}{M} \varrho \ . 
  \end{align*}
  Since $(d_k)_{k\geq 0}$ is sublinear, there exists for any $\eps>0$ a constant $M$ large enough such that $C_K^2 \varrho \frac{d_M}{M} \leq \eps/2$ concluding the first estimate of~\eqref{e:AB:weak_comp}. The second one is very similar. Again, there exists for any $\eps>0$ and $N\geq 1$ a $t_2=t_2(\eps,N)$ such that $\omega_l^\varrho\leq \bra*{1 + \frac{\eps}{2}} c_l(t)$ for all $l=0,\dots, N-1$ and all $t\geq t_2$. Then, the second part of Assumption~\eqref{e:ass:K4} implies
  \begin{align*}
    \frac{B_l[\omega^\varrho]}{B_l[c(t)]} &= \frac{\sum_{k=0}^{N-1} K(l,k) \, \omega_k^\varrho + \sum_{k\geq N} K(l,k) \, \omega_k^\varrho}{\sum_{k\geq 0} K(l,k) \, c_k(t)} \\
    &\leq 1 + \frac{\eps}{2} + \frac{ C_k \, b_l \sum_{k\geq N}(k+1)\,\omega_k^\varrho}{C_K^{-1} \, b_l \sum_{k\geq 0} c_k(t)} \leq 1 + \frac{\eps}{2} + C_K^2 \sum_{k\geq N} (k+1)\, \omega_k^\varrho \leq 1+\eps 
  \end{align*}
  for $N$ sufficiently large, such that $\sum_{k\geq N} (k+1) \omega_k^\varrho\leq \frac{\eps}{2}$. This proves estimate~\eqref{e:AB:weak_comp} completely once $t_0$ is set to $\max\set{t_1(\eps,M),t_2(\eps,N)}$.
  
  \textbf{Step 2:} The second step of the proof consists in applying the tightness estimate of Proposition~\ref{prop:tight}. Let $\gamma_l = 1$ for $0\leq l < l_0$ and
  \[
    \gamma_l = \nu_l \, \gamma_{l-1} \qquad\text{for } l\geq l_0 \ ,
  \]
  where $(\nu_l)_{l\geq 0}$ is defined in~\eqref{e:def:lambda}. The sequence $(\gamma_l)_{l\geq 0}$ satisfies the iteration
  \[
    \frac{\gamma_l}{\gamma_{l-1}} \leq \frac{\omega_l(\phi)}{\omega_{l-1}(\phi)} \frac{l^2}{(l-1)^2} \ , \qquad\text{and hence}\qquad l\, \gamma_l \leq l^{3} \omega_l(\phi) \ . 
  \]
  Since $\omega(\tilde \phi) \in \cX$ for all $\tilde \phi<\phi_c$, $\omega(\phi)$ has arbitrary high moments for any $\phi< \phi_c$ and hence $\gamma \in \cX$. Therefore, the sequence $\eta_l = \sum_{k\geq l} k \, \gamma_k$ satisfies $\eta_l \to 0$ as $l \to \infty$. So, $\eta$ is an element of the set 
  \[
    S_\nu = \set[\big]{ \lambda = (\lambda_l) : \lambda_l \geq \lambda_{l+1} \geq 0 \text{ for all $l\geq 0$ and } \lambda_l - \lambda_{l+1} \geq \nu_l\bra*{\lambda_{l-1} -\lambda_l}} \ . 
  \]
  The cumulative distribution of the initial data $\sigma_l = \sum_{k\geq l} l \, \overbar c_l$ satisfies $\sigma_l \to 0$ as $l\to \infty$. Then, according to \cite[Lemma 4]{BC88} there exists $\hat\lambda_l$ such that $\hat\lambda_l \geq \sigma_l$ for all $l$ and $\hat\lambda_l \to 0$ as $l\to \infty$. Hence, thanks to the first step, Proposition~\ref{prop:tight} can be applied to $(c(t))_{t\geq t_0}$ to conclude for any $l\geq l_0$ and all $t\geq t_0$ that
  \[
    \sum_{k\geq l} k c_k(t) \leq \hat\lambda_l \max\set*{ 1, \frac{\varrho_0}{\lambda_{l_0}}} \ .
  \]
  This shows that $\bra*{c(t)}_{t\geq 0}$ is relatively compact in $\cX$ and $c(t) \to c^\varrho$ strongly in $\cX$ as $t\to \infty$.
\end{proof}

\subsection{Proof of Theorem~\ref{thm:longtime}}\label{s:longtime}

The assertion for the case $\varrho_0 > \varrho_c$ is exactly Theorem~\ref{thm:longtime:weak}. The statement on the convergence of the free energy is a consequence of Proposition~\ref{prop:weak:cont:free:energy} and the relation~\eqref{e:RelEnt:F} implying
\begin{align*}
  \cF[c(t)] &= \cH[c(t)| \omega^{\varrho_c}] - \log Z(\phi_c) + \varrho_0 \log \phi_c \\
  &\to -\log Z(\phi_c) + \varrho_c \log \phi_c + \bra*{ \varrho_0 - \varrho_c} \log \phi_c = \cF[\omega^{\varrho_c}] + \bra*{ \varrho_0 - \varrho_c} \log \phi_c \ . 
\end{align*}
The case $\varrho_0 < \varrho_s$ follows by combining Theorem~\ref{thm:longtime:weak} with Proposition~\ref{prop:weak_to_strong}, where the statement of the free energy is now an immediate consequence of Proposition~\ref{prop:weak:cont:free:energy}.
Finally, the case $\varrho_0 = \varrho_s$ is again a consequence of Theorem~\ref{thm:longtime:weak} combined with Proposition~\ref{prop:weak_to_strong} implying that $\varrho = \varrho_0 = \varrho_c$, which by density conservation gives the strong convergence. Again the statement on the free energy is an immediate consequence of Proposition~\ref{prop:weak:cont:free:energy}. \hfill \qedsymbol

\bigskip

\appendix

\section{Lemma of de la Vallée-Poussin\texorpdfstring{ in $\cP^\varrho$}{}}\label{s:ValleePoussin}

\begin{lem}\label{lem:ValleePoussin}
  For $\varrho>0$ and $\cC = \set{\bar c^i}_{i\in I} \subset \cP^\varrho$ the following are equivalent:
  \begin{enumerate}
   \item The family $\cC$ is uniform integrable, that is
   \begin{equation}\label{e:def:uniform_integrable}
     \lim_{l\to \infty} \sup_{\bar c\in \cC} \sum_{k\geq l} (1+k) \, \bar c_k \to 0 \ .
   \end{equation}
   \item There exists a positive increasing superlinear sequence $(g_k)_{k\geq 0}$ 
  such that 
  \begin{equation}\label{e:ValleePoussin:bound}
    \sup_{\bar c \in \cC}\sum_{k\geq 0} g_k \, \bar c_k < \infty \ . 
  \end{equation}
   Moreover, $(g_k)_{k\geq 0}$ can be chosen to satisfy the bound
  \begin{equation}\label{e:ValleePoussin:condition}
   0< (k+1)\,(g_{k+1} - g_k) \leq 2 \, g_k \qquad\text{for } k\geq 0 
  \end{equation}
    \end{enumerate}
\end{lem}
\begin{proof}
  The implication (2)$\Rightarrow$(1) is a straightforward consequence of the superlinear growth of $(g_k)_{k\geq 0}$ implying for any $\bar c\in\cC$
  \[
     \sum_{k\geq l} (1+k) \, \bar c_k \leq \frac{1+l}{g_l} \sum_{k\geq l} g_k \bar c_k \leq \frac{1+l}{g_l} \sum_{k\geq 0}  g_k \bar c_k \ .
  \]
  Taking the $\sup$ over $\cC$ and letting $l\to \infty$ proves the implication.

  For the proof of (1)$\Rightarrow$(2) the construction from~\cite{CanizoDeLaValleePoussin2006} based on \cite[Theorem 22]{dellacherie1978probabilities} is modified to satisfy the condition~\eqref{e:ValleePoussin:condition}. The tail distribution of $\bar c^i$ for $i\in I$ is defined by 
  \begin{equation*}
    C_k^i = \sum_{l\geq k} (l+1)\, \bar c_l^i \qquad\text{ for } k\geq 0\ .
  \end{equation*}
  Two auxiliary increasing sequences $(a_n)_{n\geq 1}$ and $(\ell_n)_{n\geq 0}$ are defined by 
  \begin{equation*}
    a_n = \inf\set*{ k\geq 0 \,\middle|\, \sup_{i\in I} C_k^i \leq \frac{1}{n^2}} \quad\text{ for } n \geq 1 \ ,
  \end{equation*}
  and inductively $\ell_{n+1} = \max\set*{ \ell_n+1 , a_{n+1}+1}$ starting with $\ell_0 = 0$.
  Then by construction, it holds $\sup_{i \in I} C_{\ell_n}^i \leq \frac{1}{n^2}$ for $n\geq 1$ and $C_0 = 1+\varrho$.
  One more auxiliary sequence $(\varphi_k)_{k\geq 0}$ is given by 
  \begin{equation*}
    \varphi_k = n+1 \qquad\text{for } k\in [\ell_n, \ell_{n+1}) \ .
  \end{equation*}
  By construction $\varphi_k\to \infty$ as $k\to \infty$, since $\varphi_k \geq n+1$ for $k\geq \ell_n$. 
  Then, the candidate for $g_k$ is the sequence $\varphi_k \, (k+1)$. Indeed, it holds for any $i\in I$
  \begin{align}
    \sum_{k\geq 0} \varphi_k \, (k+1) \, \bar c_k^i &= \sum_{n\geq 0} (n+1) \sum_{k=\ell_n}^{\ell_{n+1}-1} (k+1) \, \bar c_k^i = \sum_{n\geq 0} \sum_{k\geq \ell_n} (k+1)\, \bar c_k^i \notag \\
    &= \sum_{n\geq 0} C_{\ell_n}^i \leq 1+\varrho + \sum_{n\geq 1} \frac{1}{n^2} < \infty \ . \label{e:ValleePoussin:p0}
  \end{align}
  To verify the condition~\eqref{e:ValleePoussin:condition}, it is necessary to regularize the sequence $(\varphi_k)_{k\geq 1}$ by defining the following interpolation: $d_0=1$ and $\Phi_0=0$ and inductively for $n\geq 0$
  \begin{align*}
    d_{n+1} &= \min\set*{ d_n , \frac{n+1 - \Phi_{\ell_n}}{\ell_{n+1} - \ell_n}, \frac{1}{\ell_{n+ 1}}} \ ;\\
    \Phi_k &= \Phi_{\ell_n} + d_n \bra*{k - \ell_n} \qquad\text{for } k\in [\ell_n , \ell_{n+1})  \ .
  \end{align*}
  The construction ensures that $\Phi_k$ is an increasing sequence with $\Phi_k \leq \varphi_k$. Hence, $g_k = \Phi_k \, (k+1) + 1$ still satisfies~\eqref{e:ValleePoussin:bound} with an additional constant $1+\varrho$ on the right hand side. Finally to show~\eqref{e:ValleePoussin:condition}, note that for $k\in [\ell_n, \ell_{n+1})$ it holds
  \begin{equation*}
    \Phi_{k+1} - \Phi_k \leq d_{n+1} \leq \frac{1}{\ell_{n+1}} \leq \frac{1}{k+1} \ .
  \end{equation*}
  Hence, the estimate~\eqref{e:ValleePoussin:bound} follows from
  \begin{equation*}
    (k+1)\,\bra*{g_{k+1} -g_k} = (k+1) \, \bra[\big]{ (k+1) \Phi_{k+1} - k \Phi_k} \leq (k+1) \, \bra*{ 1 + \Phi_k} \leq 2 \, g_k \ , 
  \end{equation*}
  where the lower bound $\Phi_k \geq 1$ was used in the last step.
\end{proof}

\subsection*{Acknowledgement}

The author enjoyed fruitful discussions with Stefan Grosskinsky on the derivation of the system as mean-field limit as well as with Stefan Luckhaus, Barbara Niethammer, and Juan Velázquez on the Becker-Döring system and related topics. The author thanks Emre Esenturk for bringing many references to his attention. Moreover, the author is incredibly grateful for the very constructive and detailed referee reports pointing out missing steps and misprints in an earlier version. 
The author acknowledges support by the Department of Mathematics I at RWTH Aachen University, where part of this work originates.

\bibliographystyle{alphaabbr}
\bibliography{bib}
\end{document}